%% file: swm-greedy-Dec-2017.tex
\documentclass[onethmnum,final]{siamltex1213}
\usepackage{amssymb,amsmath,mathtools}
\usepackage{graphicx,xspace,bm}
\usepackage[numbers]{natbib}
\usepackage{url,hyperref}
\usepackage{enumitem}
\usepackage{nicefrac}

\newtheorem{claim}{Claim}

\newcommand{\RR}{{\mathbb R}}

\newcommand{\E}{{\mathbb E}}

\newcommand{\la}{\langle}
\newcommand{\ra}{\rangle}
\newcommand{\G}{G}
\newcommand{\etal}{\emph{et al.}\xspace}
\newcommand{\LP}{\mathsf{LP}}
\newcommand{\OPT}{\mathsf{OPT}}

\newcommand{\oswm}{online SWM\xspace}
\newenvironment{proofof}[1]{\smallskip\noindent{\it Proof of #1.}}%
        {\hspace*{\fill}$\Box$\par}

\DeclareMathOperator*{\argmax}{arg\,max}

\begin{document}
\slugger{sicomp}{xxxx}{xx}{x}{x--x}
\date{}

\title{Online Submodular Welfare Maximization:\\
  Greedy Beats $\frac{1}{2}$ in Random Order\thanks{An extended abstract of this paper appeared at ACM Symposium on Theory of Computing (STOC)' 2015.}}

\author{Nitish Korula\thanks{Google Research, New
    York. nitish@google.com}
  \and Vahab Mirrokni\thanks{Google
    Research, New York. mirrokni@gmail.com}
  \and Morteza Zadimoghaddam\thanks{Google Research, New York. zadim@google.com}}

\maketitle

\begin{abstract}
  In the Submodular Welfare Maximization (SWM) problem, the input
  consists of a set of $n$ items, each of which must be allocated to
  one of $m$ agents. Each agent $\ell$ has a valuation function
  $v_\ell$, where $v_\ell(S)$ denotes the welfare obtained by this
  agent if she receives the set of items $S$. The functions $v_\ell$
  are all submodular; as is standard, we assume that they are monotone
  and $v_\ell(\emptyset) = 0$. The goal is to partition the items into
  $m$ disjoint subsets $S_1, S_2, \ldots S_m$ in order to maximize the
  social welfare, defined as $\sum_{\ell = 1}^m v_\ell(S_\ell)$.  A
  simple greedy algorithm gives a $1/2$-approximation to SWM in the
  offline setting, and this was the best known until Vondrak's recent
  $(1-1/e)$-approximation algorithm~\cite{V08}.

  In this paper, we consider the online version of SWM. Here, items
  arrive one at a time in an online manner; when an item arrives, the
  algorithm must make an irrevocable decision about which agent to
  assign it to before seeing any subsequent items. This problem is
  motivated by applications to Internet advertising, where user ad
  impressions must be allocated to advertisers whose value is a
  submodular function of the set of users / impressions they receive. 
  There are two natural models that differ in the order in which items
  arrive. In the fully \emph{adversarial} setting, an adversary can
  construct an arbitrary / worst-case instance, as well as pick the
  order in which items arrive in order to minimize the algorithm's
  performance. In this setting, the $1/2$-competitive greedy algorithm
  is the best possible. To improve on this, one must weaken the
  adversary slightly: In the \emph{random order} model, the adversary
  can construct a worst-case set of items and valuations, but does not
  control the order in which the items arrive; instead, they are
  assumed to arrive in a random order. The random order model has been
  well studied for online SWM and various special cases, but the best
  known competitive ratio (even for several special cases) is $1/2 +
  1/n$~\cite{DNS05,DS06}, barely better than the ratio for the
  adversarial order. Obtaining a competitive ratio of $1/2 +
  \Omega(1)$ for the random order model has been an important open
  problem for several years.
  We solve this open problem by demonstrating that the
  greedy algorithm has a competitive ratio of at least $0.505$ for
  \oswm in the random order model. This is the first result showing a
  competitive ratio bounded above $1/2$ in the random order model,
  even for special cases such as the weighted matching or budgeted
  allocation problems (without the so-called `large capacity'
  assumptions). For special cases of submodular functions including
  weighted matching, weighted coverage functions and a broader class of
  ``second-order supermodular'' functions, we provide a different
  analysis that gives a competitive ratio of $0.51$.  We analyze the
  greedy algorithm using a factor-revealing linear program, bounding
  how the assignment of one item can decrease potential welfare from
  assigning future items.

  In addition to our new competitive ratios for \oswm, we make two
  further contributions: First, we define the classes of
  \emph{second-order} modular, supermodular, and submodular functions,
  which are likely to be of independent interest in submodular
  optimization. Second, we obtain an improved competitive ratio via a
  technique we refer to as \emph{gain linearizing}, which may be
  useful in other contexts (see \cite{MirrokniZ15}): Essentially, we
  linearize the submodular function by dividing the gain of an optimal
  solution into gain from individual elements, compare the algorithm's
  gain when it assigns an element to the optimal solution's gain from
  the element, and, crucially, bound the extent to which assigning
  elements can affect the potential gain of other elements.
\end{abstract}


\begin{AMS}\end{AMS}

\pagestyle{myheadings}
\thispagestyle{plain}
\markboth{NITISH KORULA, VAHAB MIRROKNI AND MORTEZA ZADIMOGHADDAM}{ONLINE SUBMODULAR WELFARE MAXIMIZATION IN RANDOM
  ORDER}

\section{Introduction}
As a general abstraction of many economic resource-allocation
problems, {\em submodular welfare maximization} (abbreviated as SWM)
is a central optimization problem in combinatorial auctions and has
attracted significant attention in the research area at the
intersection of economics, game theory, and computer science.  In this
problem, an auctioneer sells a set $N$ of $n$ items to a set $M$ of
$m$ agents.  The value of agent $\ell \in M$ for any subset (bundle)
of items is given by a submodular valuation set function
$v_\ell:2^N\rightarrow \RR_+$, where $v_\ell(S)$ represents $i$'s
maximum willingness to pay for the bundle $S$. The two standard
assumptions on each $v_\ell$ (besides submodularity) are that if
$S\subseteq T$ then $v_\ell(S)\leq v_\ell(T)$ (monotonicity), and that
$v_\ell(\emptyset)=0$ (normalization). The objective is to partition
$N$ into $m$ disjoint subsets $S_1, S_2, \ldots, S_m$, and give set
$S_\ell$ to agent $\ell$ in a way that maximizes {\em the social
  welfare}, i.e. the expression $\sum_{\ell = 1}^m v_\ell(S_\ell)$.


In the online version of the problem (also known as {\em online
  submodular welfare maximization}), referred to as {\em \oswm}, items
in $N$ arrive one by one online, and upon arrival of an item, it must
be assigned immediately and irrevocably to one of the agents. That is,
the assignment of an item must be made before any subsequent items
arrive, and it may not be changed later. 

The online SWM problem is a natural generalization of the online
matching~\cite{KVV,kp-balance,FMMM09,AGKM11,MY11}, budgeted allocation
~\cite{MSVV,buchbinder-jain-naor,devanur-hayes,GM08} and online
weighted matching problems~\cite{AGKM11,FKMMP09}, along with more
general classes of online allocation / assignment
problems~\cite{FHKMS10,VeeVS,AWY09,DHKMY13}. Besides being
theoretically important, these problems have a number of practical
applications including Internet advertising, network routing, etc.
These online allocation problems have been studied in both worst-case
/ \emph{adversarial} and \emph{stochastic} settings. In the
adversarial arrival model, an adversary constructs a worst-case
instance, and can order the items arbitrarily in order to minimize the
algorithm's performance. In contrast, in the \emph{random order}
arrival model, the adversary can construct an arbitrary instance, but
the order in which items arrive is considered to be chosen uniformly
at random. Here, the performance of the algorithm is computed as the
average over the random choice of the arrival order of the items.

The Submodular Welfare Maximization problem has been studied
extensively as both an offline and an online optimization problem: For
the offline optimization problem, one can easily observe that SWM is a
special case of the monotone submodular maximization problem subject
to a (partition) matroid constraint. As a result, an old result of
Nemhauser, Wolsey, and Fisher~\cite{NWF78} implies that a simple greedy
algorithm achieves a $1/2$-approximation. Improving this approximation
factor for the offline SWM was open until Vondrak~\cite{V08} presented
a new $1-1/e$-approximation algorithm for the problem; this is the
best possible using a sub-exponential number of oracle
calls~\cite{MSV08}.

For the online SWM problem, a simple online greedy algorithm (assign
each item to the agent whose marginal valuation increases the most)
achieves a competitive ratio of $1/2$ for the adversarial
model~\cite{NWF78b,LLN}.  As online SWM and its special cases are of
considerable theoretical and practical interest, there has been a
large body of work trying to improve upon this competitive ratio in
both the adversarial and random arrival models. For example, in the
adversarial model, $1-1/e$-competitive algorithms have been achieved
for the special case of the online matching problem, as well as the
the budgeted allocation and online weighted matching problems under
the so-called large capacity assumption~\cite{KVV,MSVV,FKMMP09}.
However, such a result is not possible for the general \oswm problem
in the adversarial setting, where the simple greedy algorithm is the
best possible: A recent result by Kapralov, Post, and
Vondrak~\cite{KPV13} shows that this $1/2$-approximation is tight for
\oswm unless NP=RP.  This hardness result, however, does not rule out
improving the approximation factor of $1/2$ for the random order
model.  In fact, getting a $1/2+\Omega(1)$ approximation factor
remains an open problem even for special cases of \oswm such as the
budgeted allocation problem, and weighed matching with free
disposal~\cite{FKMMP09,DHKMY13}.  

\medskip {\bf \noindent Our Contributions and Techniques.}  In this
paper, we resolve the open problem of obtaining an improved
competitive ratio for random order \oswm:

\begin{theorem}\label{thm:general-swm}
  The Greedy algorithm has competitive ratio at least $0.5052$ for
  \oswm in the random order model.
\end{theorem}

Prior to our work, the best known algorithm for this problem (even for
special cases such as weighted matching or budgeted allocation) gave a
${1\over 2} + {1\over n}$-approximation~\cite{DNS05,DS06}.  Thus, our
result is also the first $1/2 + \Omega(1)$-competitive algorithm for
the budgeted allocation and the online weighted matching problems
under the random order model. Beating the factor of $1/2$ for the
online weighted matching and budgeted allocation problems in either
the adversarial or the random order model (without the large capacity
constraints) has remained a major open problem in the area. 

Further, for a broad class of submodular functions, we can strengthen
our analysis. To do this, we introduce the concept of
\emph{second-order supermodular} functions; to the best of our
knowledge, this class of functions has not been explicitly studied
before.

\begin{definition}
  For a submodular function $f$, let $MG(A, e) = f(A \cup \{e\}) -
  f(A)$ denote the marginal gain from adding element $e$ to set
  $A$. For sets $A, S$, we define $GR(A, S, e) = MG(A, e) - MG(A \cup
  S, e)$ as the amount by which $S$ reduces the marginal gain from
  adding $e$ to $A$. (Here, $GR$ stands for Gain Reduction.) Note that
  by definition of submodularity, $GR(A, S, e)$ is always
  non-negative. 

  \begin{itemize}
  \item The function $f$ is said to be \emph{second-order modular} if, 
    for all sets $A, B, S$ such that $A \subseteq B$, and $S \cap B =
    \emptyset$, and all elements $e$, we have: $GR(A, S, e) = GR(B, S,
    e)$.

  \item The function $f$ is \emph{second-order supermodular}
    if, for all sets $A, B, S$ such that $A \subseteq B$, and $S \cap B
    = \emptyset$, and all elements $e$, we have: $GR(A, S, e) \ge GR(B, S,
    e)$. Equivalently, $MG(A, e) - MG(B, e) \ge MG(A \cup S, e) - MG(B
    \cup S, e)$.

  \item The function $f$ is said to be \emph{second-order submodular}
    if, for all sets $A, B, S$ such that $A \subseteq B$, and $S \cap B
    = \emptyset$, and all elements $e$, we have: $GR(A, S, e) \le GR(B, S,
    e)$.
  \end{itemize}

\end{definition}

It is well known that when considering the multilinear extension of $f$,
submodularity corresponds to a non-positive second-order derivative. Our
definition of second-order submodularity implies non-positive
\emph{third-order} derivatives, while the definition of second-order
supermodularity implies a non-negative third-order derivative. 

We believe these classes are likely to be of independent interest, as
they help partition the space of submodular functions, and may refine
our understanding of submodular optimization. Several natural
submodular functions can be classified in this framework: For example,
cut functions are second-order modular, while weighted coverage
functions and weighted matching functions are second-order
supermodular. 

\begin{theorem}\label{thm:secondorder-swm}
  The Greedy algorithm has competitive ratio at least $0.5104$ for
  \oswm in the random order model if the valuation functions of agents
  are second-order supermodular functions.
\end{theorem}

Note that our results also imply simple $1/2+\Omega(1)$ approximation
algorithms for the offline SWM problems: simply permute the items
randomly and apply the online Greedy algorithm.\footnote{We note that
  improving on the ratio of $1/2$ for the offline Submodular Welfare
  Maximization Problem was an open problem for nearly 40 years until
  the result of Vondrak~\cite{V08}.} We focus on the Greedy algorithm
for several reasons: It is simple and natural, besides being easy to
implement and likely to be used in practice. Further, it is optimal in
the adversarial setting.  Algorithms that perform well in both
adversarial and stochastic settings (see, for instance,~\cite{MOZ12})
are of considerable practical utility; by showing that the Greedy
algorithm achieves a ratio better than $1/2$ in the random order
model, we provide further justification for its practical
importance. Moreover, the Greedy algorithm for online SWM has been
extensively studied in strategic settings~\cite{PLST12, SyrgkanisT12,
  SyrgkanisT13}.

Our approach to analyzing the performance of the Greedy algorithm is
to understand how item allocations interact via the technique we call
\emph{gain linearizing}: We first formulate some basic inequalities
about \emph{any} greedy allocation which yield another proof that the
Greedy algorithm achieves a competitive ratio of $1/2$. To go beyond
this ratio, we note that the assignment of an item can reduce the
expected gains from future items, but this reduction in future gains
is bounded, and the random order ensures that it is `evenly spread
out' among future items. We formalize this intuition in
Section~\ref{sec:prelims}; once we can bound such interactions, we
derive constraints for a factor-revealing Linear Program, which we
explicitly analyze to get a lower bound on the competitive ratio of
the Greedy algorithm. 

\subsection{Related Work}
The online submodular welfare maximization problem is a generalization
of various well-studied online allocation problems, including problems
with practical applications to Internet advertising. These include
online weighted $b$-matching (with free disposal), also referred to as
the {\em Display Ads Allocation}
problem~\cite{FKMMP09,FHKMS10,AWY09,VeeVS}, and budgeted allocation or
the {\em AdWords} problem~\cite{MSVV,devanur-hayes}.  In both of these
problems, a publisher must assign online impressions to a set of
agents corresponding to advertisers; the goal is to optimize
efficiency or revenue of the allocation while respecting pre-specified
contracts with advertisers.  Both of these problems have been studied
in the competitive adversarial
model~\cite{MSVV,FKMMP09,buchbinder-jain-naor} and the stochastic
random order model~\cite{devanur-hayes,FHKMS10,AWY09,VeeVS}.
 
In the online weighted $b$-matching problem with free disposal,
motivated by display advertising, we are given a set of $m$
advertisers; advertiser $j$ has a set $S_j$ of eligible impressions
and demand of at most $N(j)$ impressions. The ad-serving algorithm
must allocate the set of $n$ impressions that arrive online. Each
impression $i$ has value $w(i,j) \geq 0$ for advertiser $j$. The goal
of the algorithm is to assign each impression to at most one
advertiser, while maximizing the value of all the assigned impressions
and ensuring that advertiser $j$ does not receive more than $N(j)$
impressions.  Expressed combinatorially, this corresponds to finding a
maximum weight $b$-matching online. This weighted $b$-matching problem
is considered in~\cite{FKMMP09}, which showed that the problem is
inapproximable without exploiting {\em free disposal}. When the {\em
  demand of each advertiser is large}, a $(1-{1\over e})$-competitive
algorithm exists~\cite{FKMMP09}, and this is the best possible.
In the budgeted allocation problem, motivated by sponsored search
advertising, the ad-serving algorithm allocates impressions resulting
from search queries. Advertiser $j$ has a budget $B(j)$ on the total
spend instead of a bound $N(j)$ on the number of
impressions. Assigning impression $i$ to advertiser $j$ consumes
$w(i,j)$ units of $j$'s budget instead of 1 of the $N(j)$ slots, as in
the weighted $b$-matching problem. $1-{1\over e}$-competitive
algorithms have been designed for this problem under the assumption of
\emph{large budgets}~\cite{MSVV,buchbinder-jain-naor}.

The random order model has also been studied extensively for these
problems. In particular, a {\em dual technique} has been developed to
solve problems in this setting: This approach is based on computing an
offline optimal dual solution based on the first $\epsilon$ fraction
of items / impressions, and using this solution online to assign the
remaining vertices. Following the first such training-based dual
algorithm of \cite{devanur-hayes} for the budgeted allocation problem,
training-based $(1-\epsilon)$-competitive algorithms have been
developed for the weighted matching problem and its generalization to
various packing linear
programs~\cite{FHKMS10,VeeVS,AWY09,MolinaroRavi}.  These papers
develop $(1-\epsilon)$-competitive algorithms for online stochastic
packing problems in the random order model, if ${{\OPT}\over
  w_{ij}}\ge O({m\log n \over \epsilon^2})$ and the demand of each
advertiser / agent is large. Recently, Kesselheim
\etal~\cite{Kesselheim14} have weakened this condition using a
\emph{primal}-based technique. In a separate line of work, improved
approximation algorithms have been proposed for the unweighted online
bipartite matching problem (which is a special case of both the
weighted matching and budgeted allocation problems) in the random
order model~\cite{KMT11,MY11}, even without the assumption that the
demand of each agent is large.

\smallskip {\em \bf \noindent i.i.d. Model.}  Other than the
adversarial and random-order models studied in this paper, online ad
allocation problems have been studied extensively in the {\em
  i.i.d. stochastic models} in which impressions arrive i.i.d.\
according to a known or an unknown distribution.  A key technique in
this area is the \emph{primal approach} which is based on 
solving an offline allocation problem on the instance that we expect
to arrive according to the stochastic information, and then applying
this offline solution online.  This technique has been applied to the
online stochastic matching problem~\cite{KVV} and in the i.i.d.\ model
with known distributions~\cite{FMMM09,MOS11,HMZ11,AHL-Prophet}, and
resulted in improved competitive algorithms.
  
Interestingly, a new \emph{hybrid} technique can be applied to obtain
a $1-1/e$-competitive algorithm for the budgeted allocation problem in
the i.i.d. model~\cite{DJSW11,DevanurSA12}, and a similar ratio can be
obtained for the more general online SWM
problem~\cite{KPV13}. However, all results based on such techniques
seem to only apply to the i.i.d. model~\cite{DJSW11,KPV13} and not the
random order model.  Generalizing such results to the random order
model remains an interesting open problem in the area.

\section{Preliminaries and Key Ideas}
\label{sec:prelims}

We study the following welfare maximization problem, referred to as
\oswm: $n$ items from a set $N$ arrive sequentially online, and when
each item arrives, it should be irrevocably assigned to one of $m$
agents from the set $M$. Each agent $1 \leq \ell \leq m$ has valuation
function $v_\ell: 2^{[n]} \to \mathbb{R_{+}}$.  We assume that each
agent's valuation function $v_\ell$ is submodular, non-negative and
monotone.  The goal is to assign the items to agents in order to
maximize the social welfare, defined as $\sum_{i=1}^m v_\ell(S_\ell)$
where $S_\ell$ is the set of items assigned to $\ell$. We assume
w.l.o.g. that all of the $n$ items are distinct, but we sometimes
consider the union of two or more allocations, as a result of which
there may exist multiple copies of items: In this case, the set of
items assigned to a single agent is the union of the sets assigned to
it under the allocations, so it does not actually receive multiple
copies of the same item. (Equivalently, one could extend the
submodular function to multisets in the natural way so that the
marginal value of an agent for the second copy of an item is $0$.)

Algorithms for maximizing welfare in combinatorial auctions and
submodular optimization are required to be polynomial in the natural
parameters of the problem, $m$ and $n$. However, since the ``input''
(the valuation functions) is of exponential size, one must specify how
it can be accessed. Most works in this field have taken a ``black
box'' approach in which bidders' valuation functions are accessed via
oracles that can answer specific type of queries. The most natural and
popular oracle query model is the {\em value query
  model}~\cite{DNS05,DS06,LLN}, which we also use: the query to a
valuation function $v_\ell$ is in the form of a bundle $S\subseteq N$,
and the response is $v_\ell(S)$. 

As discussed above, no online algorithm can achieve a competitive
ratio greater than $1/2$ for \oswm in the adversarial setting, and
hence we focus on the \emph{random order} model. We measure the
competitive ratio of an algorithm as the ratio of the expected social
welfare achieved by the algorithm's assigment (where the expectation
is taken over all $n!$ permutations) to the optimum social
welfare. Our main contribution is to demonstrate that the simple and
natural Greedy algorithm achieves a competitive ratio of $1/2 +
\Omega(1)$.

What is the main difficulty in proving that the Greedy algorithm
achieves a competitive ratio better than $1/2$? It is obvious that the
expected welfare that Greedy obtains by allocating the
first item is at least as much as the optimum algorithm obtains by
allocating this item. Of course, the disadvantage of the greedy
allocation is that by assigning one item, it may severely reduce the
potential welfare gain from items that arrive in the future. Our key
insight is into how items interact and how allocating one item can
affect the welfare \emph{gain} that can be obtained from allocating
other items. 

\subsection{Understanding Item Interactions and Gain}

In order to study interactions between allocations, we need to
introduce the following notation:

\begin{definition}\label{def:allocation}
  An allocation of items to agents, denoted by $A$, consists of $m$
  subsets of items $A = \{A_\ell\}_{\ell=1}^{m}$ where $A_\ell$ is the
  set of items allocated to agent $\ell$, and every item is assigned
  to at most one agent. We note that some items may remain unallocated
  in an allocation. We denote the total welfare or value of allocation
  $A$ by $V(A) = \sum_{\ell=1}^m v_{\ell}(A_{\ell})$.
\end{definition}

We let $A^*$ be the allocation that maximizes social welfare
$OPT=\sum_{\ell=1}^m v_\ell(A^*_\ell)$ where $OPT$ is the maximum social
welfare. Without loss of generality (by normalizing all values), we
assume that $OPT$ is equal to $1$. For an item $j$, we define $opt_j$
to be the agent that receives item $j$ in the optimum allocation
$A^*$, i.e. $j$ belongs to $A^*_{opt_j}$. For a permutation $\sigma$
on the items $[n]$, we define $\sigma^i$ to be the set of first $i$
items in $\sigma$, in other words $\{\sigma_1, \sigma_2, \cdots,
\sigma_i\}$. We note $\sigma$ is defined to just fix an arbitrary ordering of the items. It does not need to be a random shuffling and it is just for the purpose of indexing the items. 

\medskip

We can now introduce a central concept in our analysis, used to
provide a lower bound on the welfare gains of Greedy. In standard
analyses of greedy algorithms, a common approach is to show that one
good option available to the greedy algorithm is the choice made by
the optimal algorithm. In our setting, this would translate to showing
that Greedy can obtain a good increase in welfare by assigning item
$j$ to $opt_j$. The other piece of this approach is to show that the
greedy choice for allocating item $j$ does not reduce the welfare gain
that can be obtained from \emph{future} items by more than the welfare
increase from $j$. Together, these two ideas imply that Greedy is
$1/2$-competitive in the worst case.  In order to formally argue about
this, we introduce the following notation to describe the marginal
gain of assigning item $j$ to agent $\ell = opt_j$ based on a current
allocation and a (partial) optimum allocation:

\begin{definition}\label{def:Gain} 
  Fix an arbitrary permutation $\sigma$ (such as the identity
  permutation).  For any item $j$ and any allocation $A$ let $\ell =
  opt_j$, and $i+1$ denote the index of item $j$ in $\sigma$. We
  define $Gain(j, A)$ as:
  $$ 
  v_\ell \left( \{j\} \cup A_\ell \cup (A^*_\ell \cap \sigma^i) \right) - v_\ell\left(A_\ell
  \cup (A^*_\ell \cap \sigma^i)\right).
  $$
\end{definition}

That is, $Gain(j, A)$ denotes the marginal gain we get from assigning
$j$ to the agent $\ell$ that receives it in the optimal solution,
assuming that $\ell$ has already received all items in $A_\ell$ based
on the allocation $A$, \emph{as well as} those of the first $i$ items
(under $\sigma$) that the optimal solution allocates to $\ell$
(i.e. $A^*_\ell \cap \sigma^i$). It is important to note that in the
definition, the permutation $\sigma$ is fixed and the same permutation
is used to define $Gain(j, A)$ for all items $j$ and allocations $A$.

Intuitively, $Gain(j, A)$ captures the further marginal gain one can
achieve from item $j$, given allocation $A$. As a simple example,
consider the case when $A$ is an empty allocation (no item is
allocated in $A$). In this case, the sum $\sum_{j=1}^n Gain(j,A)$ is
equal to $OPT$ because for any permutation $\sigma$, the sum
$\sum_{j=1}^n Gain(j, A)$ of marginal gains of the $n$ items obtained
by the optimal allocation is the definition of $OPT$. For any set $S$,
we use $Gain(S, A)$ to denote $\sum_{j \in S} Gain(j, A)$.

How is the concept of $Gain$ useful?  Let $\pi$ denote the arrival
order of the items. We denote the allocation of Greedy on the first
$i$ items (that is, on the sequence $\pi^i$) by $A^i$.  The following
lemma captures the marginal welfare Greedy achieves at each step by
using the notion of $Gain$ variables.

\begin{lemma}\label{lem:GainLB} 
  If item $j$ arrives at position $i+1$ under permutation $\pi$, the
  increase in welfare that Greedy achieves by allocating this item is
  at least $Gain(j, A^i)$. 
\end{lemma}
\begin{proof}
  Recall that Greedy's allocation on the first $i$ items is $A^i$. We
  now consider how much Greedy gains by allocating item $j$ in the
  next position.  Since Greedy allocates the item in order to
  maximize marginal gain, its gain on $j$ is at least as much
  as item $j$'s marginal gain on agent $\ell = opt_j$ which is equal
  to $v_{\ell}(\{j\} \cup A^i_{\ell}) - v_{\ell}(A^i_{\ell})$.
  (Recall that $A^i_{\ell}$ is the set of items assigned to agent
  $\ell$ in allocation $A^i$.)

  Note that this lower bound is similar to the expression in the
  definition of $Gain(j,A^i)$ (it only lacks the $A^*_{\ell} \cap
  \sigma^i$ parts). As a direct result of submodularity of the
  function $v_{\ell}$, we have:
  \[
  v_{\ell}(\{j\} \cup A^i_{\ell}) - v_{\ell}(A^i_{\ell}) \geq
  v_{\ell}(\{j\} \cup A^i_{\ell} \cup (A^*_{\ell} \cap \sigma^i)) -
  v_{\ell}(A^i_{\ell} \cup (A^*_{\ell} \cap \sigma^i))
  \]

  \noindent That is, $v_\ell((\{j\} \cup A^i_{\ell}) - v_{\ell}(A^i_{\ell}) \ge
  Gain(j, A^i)$ for any permutation $\sigma$.
\end{proof}

What does this have to do with interactions between items? For a fixed
item $j$, the value $Gain(j,A^i)$ is non-increasing in $i$.  In other
words, when Greedy assigns a new item $\pi_{i}$, it may decrease the
$Gain$ values for some items.  We now prove that the total decrease in
$Gain$ values for all $n$ items is at most the increase in welfare
that Greedy obtained by allocating item $\pi_{i}$.  This next lemma
allows us to keep track of changes in these $Gain$ values throughout
the algorithm.

\begin{lemma}\label{lem:GainReduction} 
  Greedy's increase in welfare from item $\pi_i$ is at least $\sum_{j
    \in [n]} Gain(j,A^{i-1}) - Gain(j,A^i)$.
\end{lemma}
\begin{proof}
  Suppose Greedy assigns item $\pi_i$ to agent $\ell$. For each item
  $j$ with $opt_j \neq \ell$, the $Gain$ variable is intact after this
  allocation, i.e. $Gain(j,A^{i-1}) = Gain(j,A^i)$.  For those items
  $j$ with $opt_j=\ell$, for any permutation $\sigma$, the sum of
  $Gain(j, A^i)$ is equal to $v_{\ell}(A^*_{\ell} \cup A^i_{\ell}) -
  v_{\ell}(A^i_{\ell})$. A similar claim holds for $i-1$. We conclude
  that:
  \begin{eqnarray*}
    \sum_{j \in [n]} Gain(j,A^{i-1}) - Gain(j,A^i) =
  v_{\ell}(A^*_{\ell} \cup A^{i-1}_{\ell}) - v_{\ell}(A^{i-1}_{\ell})
  &-& \left( v_{\ell}(A^*_{\ell} \cup A^i_{\ell}) - v_{\ell}(A^i_{\ell})
  \right) \\
  = v_{\ell}(A^i_{\ell}) - v_{\ell}(A^{i-1}_{\ell}) + v_{\ell}(A^*_{\ell}
  \cup A^{i-1}_{\ell}) - v_{\ell}(A^*_{\ell} \cup A^i_{\ell}) &\leq&
  v_{\ell}(A^i_{\ell}) - v_{\ell}(A^{i-1}_{\ell}).
  \end{eqnarray*}

  The final inequality follows from the fact that the difference of
  the two dropped terms is negative due to the monotonicity of
  $v_\ell$.  The final expression is precisely the marginal gain
  Greedy achieves on item $\pi_i$, which completes the proof.
\end{proof}

Lemmas~\ref{lem:GainLB} and \ref{lem:GainReduction} (note that neither
uses the fact that we have a random permutation) use the concept of
$Gain$ to give an alternate proof of the simple fact that the Greedy
algorithm is 1/2-competitive in the adversarial model: We start with
an empty allocation $A^0$, and Greedy assigns items one by one. At the
beginning, the total $Gain = Gain(N,A^0)$ is equal to $OPT$. Let $G$
denote the sum of the final gains $Gain(N,A^n)$.
Lemma~\ref{lem:GainReduction} implies that the value of the allocation
$A^n$ is at least $OPT - G$. On the other hand, Lemma~\ref{lem:GainLB}
shows that if item $j$ arrives at position $i+1$, the algorithm's
marginal increase in welfare is at least $Gain(j, A^i) \ge Gain(j,
A^n)$. Therefore, the total welfare obtained by the algorithm is at
least $G$. Since the welfare is at least $\max\{G, $ $OPT-G\}$, it must
be at least $OPT/2$.

\subsection{Going Beyond $1/2$: Our Techniques}
\label{subsec:beyond_half}

Of course, we wish to argue that Greedy does \emph{better} than half
in the random order model. When we applied Lemma~\ref{lem:GainLB} in
the previous paragraph, we could have lower bounded Greedy's marginal
increase in welfare from allocating item $j$ in position $i+1$ by
$Gain(j, A^i)$; instead, we used the weaker lower bound $Gain(j, A^n)
\le Gain(j, A^i)$. When this inequality is not tight (that is, if the
gain of item $j$ is reduced \emph{after} it arrives by subsequent
Greedy allocations), we can obtain an improved competitive ratio.
To make this formal, we define $\beta$ as the
expected value of total amount by which the Greedy allocation reduces the gain of items
which have \emph{already} arrived. Formally $\beta$ is $\sum_{j \in N} \E[Gain(j, A^{t_j-1}) - Gain(j, A^n)]$ where $1 \leq t_j \leq n$ is the arrival time of item $j$, and $A^{t_j-1}$ and $A^n$ respectively denote the allocation of algorithm Greedy right before $j$ arrives and by the end of algorithm.

\begin{lemma}\label{lem:betaLB}
  The competitive ratio of Greedy is at least $1/2 + \beta/2$.
\end{lemma}
\begin{proof}
  If item $j$ arrives in position $i+1$, define $h_j$ as the reduction
  in the gain of $j$ \emph{after} it arrived; that is, $h_j = Gain(j,
  A^{i}) - Gain(j, A^n)$.  By definition, $\beta = \sum_{j=1}^n h_j$.
  From Lemma~\ref{lem:GainLB}, the increase in Greedy's welfare from
  allocating item $j$ is at least $Gain(j, A^i) = Gain(j, A^n) +
  h_j$. Denoting $Gain(N, A^n)$ as $G$, we conclude that the total
  welfare of Greedy is at least $G + \beta$.
  Lemma~\ref{lem:GainReduction} implies that the total welfare of
  Greedy is at least $1 - G$ (recall that we normalized $OPT$ to
  $1$. Hence, Greedy's total welfare is at least $\max\{G + \beta, 1 -
  G\} \geq 1/2 + \beta/2$.
\end{proof}

If the gains of items were consistently reduced by a substantial
amount \emph{after} they arrived, we would have a proof that Greedy
has competitive ratio $1/2 + \Omega(1)$. But perhaps in the worst-case
instances, Greedy's allocation of the $i$th item only reduces $Gain$
values for future items? By taking advantage of the random arrival
order, we can understand how the reduction in $Gain$ is distributed
among future items. In particular, the \emph{second} item to arrive is
unlikely to have its $Gain$ reduced significantly (in expectation) by
the first item, but when the \emph{last} item arrives, it is quite
likely that it has a very small $Gain$ value, because some previous
(mis)-allocations of the greedy algorithm mean that assigning it to
its optimal agent produces little value. 

A straightforward differential equation analysis based on
Lemmas~\ref{lem:GainLB} and \ref{lem:GainReduction} (which ultimately
yields a competitive ratio of $1/2$) shows that when items arrive in a
random order, the expected welfare from greedily assigning the first
$t$ fraction of items is at least $w(t) = t-\frac{t^2}{2}$ that of the
optimal.  Assuming this analysis is tight, the expected rate at which
welfare increases is $w'(t) = 1-t$ which is essentially \emph{zero}
for the last few items (as opposed to at least $1$ for the first
item). The key contribution we make, then, is to go beyond such
analyses: We derive new lower bounds on the welfare obtained from
items that arrive towards the end of the order. This allows us to show
that the total welfare obtained by the greedy algorithm is strictly
better than $1/2$ of the optimal.

Proving such lower bounds is non-trivial, and the main technique we
apply is lower bounding the welfare obtained from a single item by
both \emph{its own} $Gain$ when it arrives (as in
Lemma~\ref{lem:GainLB}), and \emph{how it reduces the $Gain$ of other
  items} (as in Lemma~\ref{lem:GainReduction}). In particular, we apply
these lower bounds to items arriving at the end of the random order,
which allows us to show that Greedy obtains non-zero welfare even at
the end.
In this paper, we take two distinct approaches to obtain the needed
lower bounds on $Gain$ values from items that arrive at the end of the
random order. Informally speaking, we do the following:
\begin{itemize}
\item First, we consider the special case of second-order supermodular
  functions.  We show that Greedy gets good welfare from items that
  arrive towards the end by proving that there exists an allocation of
  these items that significantly reduces the $Gain$s of the first
  $n/2$ items. (Recall from Lemma~\ref{lem:GainReduction} that
  Greedy's welfare from allocating an item is at least as much as its
  ability to reduce the $Gain$s of other items.) Our proof proceeds in
  several steps: We first show that there is an allocation of the
  first $n/2$ items that significantly reduces the $Gain$s of the last
  $n/2$ items; by symmetry, there must exist an allocation $\hat{A}$
  of the last $n/2$ items that significantly reduces the $Gain$s of
  the first $n/2$. Next, we use second-order supermodularity to argue
  that a subset of the items that arrive at the end will be able to
  reduce the $Gain$s of the first $n/2$ items at least in proportion
  to the size of this subset. More formally, we show that for items
  that arrive at the end, the only way they might not significantly
  reduce the $Gain$s of the first $n/2$ items is if previous
  allocations among the last $n/2$ items already reduced these $Gains$
  significantly. But in this case, Greedy's allocations reduced the
  gains of items that arrived previously, which means that $\beta$ is large, and Lemma~\ref{lem:betaLB} shows that Greedy
  must then have a competitive ratio better than $1/2$. We precisely
  quantify this improvement over $1/2$ by formulating a
  factor-revealing Linear Program, and analyzing it explicitly.
\item Second, we consider the general submodular valuation case. The
  previous approach no longer applies, as without the second-order
  supermodularity property, allocating a single random item might
  significantly affect the collective gains of other items. Therefore,
  we consider the (simulated) effect of assigning the last quarter of
  items \emph{three} times in succession, using three different
  allocations. Intuitively, this can be understood as showing three
  distinct allocations for these items and arguing that the
  (mis)-allocations of Greedy cannot harm all of these
  simultaneously. Therefore, these allocations together obtain large
  welfare, and the first time these items arrive (corresponding to the
  `real' arrival and allocation) is the best of these.
\end{itemize}

With these lower bounds on $Gain$ values, we can prove
Theorems~\ref{thm:general-swm} and \ref{thm:secondorder-swm} by
constructing factor-revealing linear programs; the fact that items
arrive in a random order allows us to constrain the values of the
variables representing the interactions in item $Gain$s. We then
formally analyze these linear programs, which requires some intricate
computations. Though these calculations are problem-specific, we
believe that our main technique of \emph{gain linearizing} and
bounding item interactions will be generally useful to improve
analyses of submodular allocations and greedy approaches for
random-order problems. (For one example, see \cite{MirrokniZ15}.)

\subsection{Formulating the Factor-Revealing LP}
\label{subsec:formulate_LP}

Recall from Lemmas~\ref{lem:GainLB} and \ref{lem:GainReduction} that
when item $j$ in position $i$ arrives, Greedy has marginal welfare at
least $Gain(j, A^{i-1})$. Further, assigning this item $j$ reduces the
$Gain$ values of other items (both those that have already arrived,
and those which will arrive after $j$), but the total reduction is at
most the marginal welfare of allocating item $j$. 

\begin{definition}
  \label{def:variables}
  For each index $1 \leq i \leq n$, we define $w_i$ to be the
  expected increase in welfare Greedy achieves by assigning the item
  in position $i$ (that is, the $i$th item in the arrival order). 
	Formally, we have:
	$$
	w_i = \E[V(A^{\G}(\pi^i)) - V(A^{\G}(\pi^{i-1}))]
	$$
where $\pi$ is the permutation denoting the arrival order of items, and $A^{\G}(S)$ denotes the allocation of Greedy
algorithm on the sequence of items $S$.
  Assigning item $i$ possibly reduces the $Gain$ values of items.  We
  partition this effect into two parts; for any item $i$, we have two
  variables $b_i$ and $a_i$ defined as follows to capture the
  reduction in $Gain$ values of other items:
  \begin{itemize}
  \item $b_i$ (we use $b$ to denote \emph{before}) is the expected
    reduction in $Gain$ of items $j$ that have already arrived. In formal notation, we define:
    $$
    b_i = \E[\sum_{j \in \pi^i} Gain(j,A^{\G}(\pi^{i-1})) - Gain(j,A^{\G}(\pi^i))]
    $$
    From
    our definition of $\beta$ in Section~\ref{subsec:beyond_half},
    $\beta = \sum_{i=1}^n b_i$.
  \item $a_i$ ($a$ denotes \emph{after}) is the expected reduction in
    $Gain$ of items $j$ that are going to arrive later:
     $$
    a_i = \E[ \sum_{j \in \pi^n \setminus \pi^i} Gain(j,A^{\G}(\pi^{i-1})) - Gain(j,A^{\G}(\pi^i))]
    $$
   
  \end{itemize}
  Clearly, allocating the item in position $i$ reduces the total gain
  values of all items by $b_i+a_i$ in expectation.
\end{definition}
  
 Variables $a_i, b_i,$ and $w_i$ (and similar variables defined later in the paper) 
 are all non-negative. In formulating 
 the linear programs we do not mention these non-negativity properties
 as extra constraints. We just treat them as boundary conditions. 
By definition, Greedy's performance is simply $\sum_{i=1}^n
w_i$. Our factor-revealing linear programs will consist of
three separate lower bounds on $w_i$. Two of these are common to both
the general submodular case as well as the case of second-order
supermodular functions. 

For the first constraint of our LP, we use the notation of
Definition~\ref{def:variables} to express
Lemma~\ref{lem:GainReduction} as $w_i \geq b_i+a_i$. (This follows
because $w_i$ is the increase in welfare that Greedy achieves from the
$i$th item, and $b_i+a_i$ is how much this segment reduces the $Gain$
values of all items.)

\bigskip
In the next lemma, we prove another lower bound on $w_i$, yielding the
next constraint of our linear program. Recall that we normalized our
instance so that $OPT = 1$. 

\begin{lemma}\label{lem:wiai} For any $1 \leq i \leq n$, if $OPT =
  1$, we have $w_i \ge \frac{1}{n}  - \sum_{j=1}^{i-1}
  \frac{a_j}{n-j}$ where variables $w_i$ and $a_j$ are defined in Definition~\ref{def:variables}.
\end{lemma}

\begin{proof} 
Let $x$ be the item that arrives at time $i$. 
  From Lemma~\ref{lem:GainLB}, the increase in welfare of the
  algorithm from $x$ is at least its $Gain$ value at the
  time of its arrival (time $i$); this can be written as the difference of the
  \emph{original} $Gain$ value of $x$ (before the algorithm
  starts) and the reduction in its $Gain$ value due to allocating items
  that arrived previously. Its original $Gain$ value is $Gain(x, \emptyset)$, and since $x$ is chosen uniformly at random among all items, we have $\E[Gain(x, \emptyset)] = \frac{1}{n}\sum_{y \in N} Gain(y, \emptyset) = OPT/n = 1/n$ because the sum of original $Gain$
  values is equal to $OPT=1$. 

  Now it suffices to prove that the expected reduction in the $Gain$
  value of the $i$th item does not exceed $\sum_{j=1}^{i-1}
  \frac{a_j}{n-j}$. Since we are bounding the reduction in the $Gain$
  value of an item before it arrives, we can ignore any reduction that
  occurs after it arrives. If the $Gain$ value of the $i$th item is
  reduced before it arrives, it must have been reduced by the item in
  position $j$ for some $j < i$, and hence this must have been
  included in $a_j$. 
  However, the variable $a_1$, for instance, captures reductions in
  $Gain$ values for \emph{all} items appearing after the first, not
  just the $i$th item. In general, for some $1 \le j < i$, $a_j$
  accounts for reduction in $Gain$ values of items which arrive in any
  of the positions $j+1, j+2, \cdots, n$. The probability that any of
  these items falls in position $i$ and hence the reduction of its
  $Gain$ value becomes relevant in this lower bound is $1/(n-j)$
  because there are $n-j$ future segments, and we have a random
  permutation of items. Therefore, for each $1 \leq j <i$, we should
  deduct $\frac{a_j}{n-j}$ in this lower bound.
\end{proof}

We now have two constraints that lower bound $w_i$, but if we only use
these two, we will not be able to show that Greedy has a competitive
ratio better than $1/2$. As described above, the further ingredient
required is to show that Greedy obtains good welfare from the items
that arrive towards the end of the sequence. In
Section~\ref{sec:secondOrder}, we do this for the special case of
second-order supermodular functions, proving
Theorem~\ref{thm:secondorder-swm}. Then, in
Section~\ref{sec:generalCase}, we apply a more general technique to
obtain a slightly weaker result for arbitrary submodular functions,
proving Theorem~\ref{thm:general-swm}. 

\section{Second-Order Supermodular Functions}
\label{sec:secondOrder}

\input{second-order}

\section{General Submodular Functions}
\label{sec:generalCase}

\input{general-case}

\section{Conclusions and Open Problems}
\label{sec:conclusions}

As we have seen, the Greedy algorithm, which achieves an optimal
competitive ratio of $1/2$ for \oswm in the adversarial setting, does
strictly better in the random order setting. We showed that the
competitive ratio of this algorithm is at least $0.5052$, and defined
the new and interesting class of \emph{second-order} supermodular
functions (including weighted matching and weighted coverage
functions), for which the ratio is at least $0.5104$.
This work motivates several open problems,
which are interesting directions for future research:

\begin{itemize}
\item First, we believe it should be possible to improve on the competitive
  ratios of both Theorems~\ref{thm:general-swm} and
  \ref{thm:secondorder-swm}. Our work broke the barrier of $1/2$, but
  further improvements may be possible via a more careful analysis. 
\item A natural question is whether the Greedy algorithm does in fact
  achieve a ratio of $1 - 1/e$ in the random order model. A hardness
  result showing that this ratio is impossible would be extremely
  interesting, yielding one of the first provable separations between
  the random order and i.i.d. models.
\item Finally, the new classes of second-order modular, second-order
  supermodular and second-order submodular functions that we defined
  are likely to be of independent interest. We may be able to refine
  our understanding of submodular optimization by determining which
  problems become more tractable for submodular functions in these
  classes.
\end{itemize}

\bibliographystyle{siam} 
\bibliography{swmbib}
\end{document}

%% file: second-order.tex
In this section, we prove Theorem~\ref{thm:secondorder-swm} which lower bounds the competitive ratio of 
  the Greedy algorithm for
  \oswm in the random order model if the valuation functions of agents
  are second-order supermodular functions. We first introduce
the factor-revealing linear program $\LP^{\beta}$, and prove that its solution is a valid lower bound on the competitive ratio of the Greedy algorithm. We note that $\beta$ is defined in Subsection~\ref{subsec:beyond_half} as the expectation of total amount by which the Greedy allocation reduces
the gain of items which have already arrived. We note that $\beta$ is well defined for each weighted bipartite graph with a set of items and agents as its two sides. In particular $\beta$ does not depend on the arrival order  of items as this random permutation is incorporated in the expectation that defines it. We also note that by normalizing $OPT=1$, we can assume $0 \leq \beta \leq 1$. We treat $\beta$ as a parameter in the following linear program, and not a variable.

We show in Corollary~\ref{cor:LowerBoundMPBeta} of  Subsection~\ref{subsec:analyze_LP} that 
 the value of the optimum solution of linear program $\LP^{\beta}$ is at least $0.5312-\beta$. We also know the competitive ratio is at least $\frac{1}{2} + \frac{\beta}{2}$ using Lemma~\ref{lem:betaLB}. So we conclude that 
 on any instance of \oswm, Greedy achieves a
competitive ratio of at least $\max\{\frac{1}{2} + \frac{\beta}{2},
0.5312-\beta\}$ for some $\beta > 0$. Therefore, Greedy is always at
least $\frac{1}{2} + \frac{0.0312}{3} = 0.5104$-competitive, which
proves Theorem~\ref{thm:secondorder-swm}.

Let us start with describing $\LP^{\beta}$.

\vspace{0.2in}

\begin{tabular}{|lll|} 
\hline Linear program $\LP^{\beta}$ && \\
\hline Minimize $\sum_{i=1}^n w_i$ && \\ \hline 
Subject to: && \\
$w_i \geq b_i + a_i$ & $\forall 1 \leq i \leq n$ & $(1)$\\
$w_i \geq 1/n  - \sum_{j=1}^{i-1} \frac{a_j}{n-j}$ & $\forall 1 \leq i \leq n$ & $(2)$\\
$w_i \geq \frac{\sum_{j=1}^{n/2} \left(a_j \frac{j}{n-j} - b_j\right)}{n/2} - g_{i}$ & $\forall n/2 < i \leq n$ & $(3)$ \\ 
$\beta \ge \sum_{i=1}^{n/2} b_i + \sum_{i=\frac{n}{2}+1}^ng_{i}$ & & $(4)$ \vspace{0.5mm}\\ 
 \hline
\end{tabular} \vspace{0.2in}

Linear program $\LP^{\beta}$ consists of $3n$ non-negative variables $\{a_i, b_i, w_i\}_{i=1}^n$ and $n/2$ non-negative variables $\{g_i\}_{n/2+1}^n$ where the non-negativity constraints are removed to have a concise description. Following, we show 
that for any instance of \oswm problem in random order model, based on the solution of Greedy algorithm, we can find a feasible solution for $\LP^{\beta}$ with objective value $\sum_{i=1}^n w_i$ equal to the competitive ratio of Greedy algorithm. Consequently, we can say that competitive ratio of Greedy is lower bounded by solution of $\LP^{\beta}$. 

\begin{lemma}
The solution of linear program $\LP^{\beta}$ lower bounds the competitive ratio of Greedy algorithm for any instance of \oswm in random order model with $n$ online items. Parameter $\beta$ is defined in Subsection~\ref{subsec:beyond_half}. 
\end{lemma}

Given an instance of \oswm and allocation of its items by Greedy algorithm, we can set $w_i, a_i$ and $b_i$ as defined in Section~\ref{sec:prelims}. 
Without loss of generality (by normalizing all values), we
assume that $\OPT$ is equal to $1$. Therefore the competitive ratio of Greedy algorithm, $\frac{\sum_{i=1}^n w_i}{\OPT}$ is equal to the objective function of $LP^{\beta}$ i.e. $\sum_{i=1}^n w_i$. 
 As described in
Section~\ref{sec:prelims}, the first two constraints of $\LP^{\beta}$ are direct implication of 
Lemmas~\ref{lem:GainReduction} and \ref{lem:wiai}. 
In the rest, we present how to set variables $g_i$ to find a feasible solution satisfying constraints $(3)$ and $(4)$. Before providing a formal proof of these constraints, we give an intuition on how these constraints are proved and they help us beat the $\nicefrac12$ barrier. 

As described in Section~\ref{sec:prelims}, the
straightforward analysis that obtains a competitive ratio of $1/2$
shows that Greedy's welfare from assigning the first $t$ fraction of
the items is $\ge t - t^2/2$. Therefore, at $t=0.8$, we have already
obtained welfare $0.48$, and at $t=0.9$, we have welfare $0.495$. That
is, the remaining welfare from assigning the last $0.2$ and $0.1$
fraction of the items is $0.02$ and $0.005$ respectively. To improve
on the competitive ratio of $1/2$, it is crucial to show that Greedy achieves a 
better welfare from items that arrive at the end. 
The purpose of constraint $(3)$ of $\LP^{\beta}$ is indeed to lower bound Greedy's welfare from assigning item 
$i$ for any item that arrives in the second half of the input, i.e. $i > n/2$.
This lower bound is in terms of $a_j$ and $b_j$  variables with some slack variable $g_i$. Without putting any restriction on $g_i$, one could set $g_i$ to $\infty$ and make constraints $(3)$ ineffective. Therefore, we prove constraint $(4)$ to upper bound the aggregate values of $g_i$. 
We start by defining some notation to formalize our proofs.

\begin{definition}
  Let $S_1$ be the sequence of items at position $\le n/2$, $S_2$ the
  sequence of items at position $i$ such that $n/2 < i \leq 3n/4$, and
  $S_3$ the sequence of items at position $ i > 3n/4$.  All of $S_1$,
  $S_2$, and $S_3$ are random variables (sequences).

  Let $A^{\G}(S)$ be the allocation of items in sequence $S$ of items by the Greedy
  algorithm, where we assume that no other item has arrived
  prior to $S$, and then Greedy allocates items of sequence $S$ one by one. We
  also define $A^*(S)$ to be the optimum allocation of items in $S$.
\end{definition}

We first intuitively show that Greedy's allocation of items in $S_1$ reduces the
$Gain$s of items in $S_2 \cup S_3$ by at least $1/4$.
 In other words, the expected value of $Gain(S_2 \cup S_3, \emptyset) - Gain(S_2 \cup S_3, A^{\G}(S_1))$ is at least $1/4$. For the sake of providing this intuition, we assume that the $1/2$-competitive ratio analysis is tight. Therefore, Greedy achieves a welfare of $t - t^2/2$ from the first $0 \leq t \leq 1$ fraction of items. 
If we simulate the arrival of items as a continuous process that happens at time interval $[0,1]$ (every item arrival covering a time interval of length $1/n$), we can say that  items that arrive at time $t$ increase the welfare at the rate $1-t$. In particular, items that arrive at small time interval $[t, t+ dt)$ approximately increase welfare by $(1-t)dt$. Since the $1/2$ competitive ratio analysis is tight, all this welfare increase should decrease the $Gain$ variables of other items and in fact only items that come in future. Knowing the items arrive in random order, we imply that $\frac{1/2}{1-t} (1-t)dt = \frac{dt}{2}$ is deducted from $Gain$ variables of items in $S_2 \cup S_3$ for any $t < 1/2$. Aggregating these deductions for all items in $S_1$, we have a total $Gain$ reduction of $\int_0^{1/2} \frac{dt}{2} = \frac{1}{4}$ from items in $S_2 \cup S_3$ by Greedy allocation of items in $S_1$.  
%

We note that $S_1$ is the first half of items and $S_2 \cup S_3$ is the second half of items. Given the random arrival order assumption, these two sets are symmetric. Therefore, there exists an allocation of $S_2 \cup S_3$ that reduces the $Gain$s
of $S_1$ by $1/4$ in expectation. However, Greedy's allocation of items in $S_1$ may
also reduce the $Gain$s of other items in $S_1$. We can compute this quantity similarly. 
For any $t < 1/2$, items in time interval $[t, t+dt)$ deduct the $Gain$ of items in $S_1$ by $\frac{1/2 - t}{1-t} (1-t)dt = (1/2 - t)dt$. Aggregating them yields the total reduction of $\int_0^{1/2} (1/2-t)dt = 1/8$.
Hence,
the further reduction of $Gain$s of $S_1$ by allocating $S_2 \cup S_3$
is at least $1/4 - 1/8 = 1/8$. This amount of further reduction is a concept that plays a major role in our analysis. We will capture it by defining variable $X$ in equation~\ref{eq:DefX} and lower bounding it formally in Lemma~\ref{lem:SecondHalf}.

We then take advantage of second-order
supermodularity to argue that the ability of a random subset of $S_2
\cup S_3$ to reduce the $Gain$s of items in $S_1$ is proportional to
its size. That is, a subset of $S_2 \cup S_3$ of size $k$ will be able
to reduce the $Gain$ of $S_1$ by at least $\frac{k}{n/2}
\frac{1}{8}$. In particular, the last $0.1$ fraction of the items will
be able to reduce the $Gain$ of $S_1$ by at least $\frac{0.1}{0.5}
\cdot \frac{1}{8} = 0.025$; and hence there is an allocation of these
items which could obtain welfare at least $0.025$. (Note that this is
$5$ times better than the naive analysis of Greedy.) The only reason
for Greedy to not get such a welfare increase would be if previous
allocations of items in $S_2 \cup S_3$ already reduced the $Gain$
of $S_1$ significantly. But in this case, we use
Lemma~\ref{lem:betaLB} to show that Greedy has a competitive ratio
better than $1/2$.
More formally, instead of considering an explicit fraction such as the
last $0.1$ fraction of the items, we quantify the improvement over
$1/2$ by applying this idea to write constraint $(3)$ of $\LP^{\beta}$ for every $i$.

\medskip
\paragraph{\bf{Setting variables $g_i$}}
We are ready to set variables $g_i$ based on Greedy algorithm's allocation. 
We need three main building blocks to be able to set variables $g_i$: a partial allocation $\hat{A}$, random variable $X$, and also random variables $Y_i$ for any $i > n/2$.

Let $\hat{A}$ denote the allocation of items in $S_2 \cup S_3$
that maximizes the reduction in $Gain$ of items in $S_1$ after Greedy's assignment of items in $S_1$, that is:

$$\hat{A} \coloneqq \argmax_{A \mbox{ an allocation of items in }  S_2 \cup S_3} Gain\left(S_1, A^{\G}(S_1)\right) - Gain\left(S_1, A^{\G}(S_1) \cup
\hat{A}\right)$$

Subsequently, we fix this allocation $\hat{A}$ which is a random variable (allocation) since $S_1$ is a random sequence.

Recall that $A^{\G}(S_1)$ is the allocation of Greedy on the items in
$S_1$. Let $A^{\G}_i$ denote Greedy's allocation on the first $i$
items from the total sequence of items $(S_1, S_2, S_3)$; therefore, we
have $A^{\G}_{n/2} = A^{\G}(S_1)$. Let $\hat{A}_{i+}$ denote the
allocation of $\hat{A}$ on items $i$ through $n$; that is,
$\hat{A}_{(\frac{n}{2}+1)+} = \hat{A}$. For any $\frac{n}{2}+1 \le i \le n$,
$A^{\G}_{i-1} \cup \hat{A}_{i+}$ denotes an allocation of all $n$
items.

Let $X$ denote the further reduction in the $Gain$ of items in $S_1$
from allocating items in $S_2 \cup S_3$ according to $\hat{A}$. 

\begin{equation}\label{eq:DefX}
X = Gain\left(S_1, A^{\G}(S_1)\right) - Gain\left(S_1, A^{\G}(S_1) \cup
\hat{A}\right)
\end{equation}

We also define $Y_i$ for $i > n/2$ as the reduction in $Gain$ of items in
$S_1$ by allocating the items with indices $i$ through $n$ according
to $\hat{A}_{i+}$, assuming that items with indices $1$ through $i-1$
were allocated according to $A^{\G}_{i-1}$.  Note that $X = Y_{n/2 + 1}$
by definition, but we use the shorter notation for easier reading. Formally,

$$ Y_i = Gain\left(S_1,
A^{\G}_{i-1}\right) - Gain\left(S_1, A^{\G}_{i-1} \cup \hat{A}_{i+}\right)
$$

For each $i > n/2$, we let random variable
$G_i$ denote $\frac{X}{n/2} - \frac{Y_i}{(n-i+1)}$, and let $g_i$
denote $\E[G_i] = \frac{\E[X]}{n/2} - \frac{\E[Y_i]}{(n-i+1)}$. Next we show constraints $(3)$ and $(4)$ hold with this setting of variables $g_i$.
Constraint $(3)$ is a direct implication of lower bounding $w_i$ in terms of $\E[Y_i]$ (Proposition~\ref{prop:wiyi}), and also lower bounding $\E[X]$ in terms of variables $a_j$ and $b_j$ (Lemma~\ref{lem:SecondHalf}). 

By definition of $Y_i$, we know there is a way to
allocate the last remaining $n-i+1$ items to reduce the $Gain$ values of items in $S_1$ by $Y_i$ (using allocation $\hat{A}_{i+}$).
Applying Lemma~\ref{lem:GainReduction} implies that the same allocation obtains  welfare at least $Y_i$. Note that the fact that we are allocating these items greedily implies that the first of these items obtains at
least the `average' welfare of any allocation. This lower bound is formalized in the following proposition.
 \begin{proposition}\label{prop:wiyi}
  For all $i > n/2$, $w_i \ge \E[Y_i] / (n-i+1)$.
\end{proposition}

So to complete the proof of Constraint $(3)$, it suffices to lower bound $\E[X]$ as follows. 

\begin{lemma}\label{lem:SecondHalf} 
  $\E[X] \ge \sum_{j=1}^{n/2}a_j\frac{n/2}{n-j} - \sum_{j=1}^{n/2}
    \left(a_j \frac{n/2-j}{n-j} + b_j \right)$
  which can be simplified to
  $\sum_{j=1}^{n/2}\left(a_j \frac{j}{n-j} - b_j \right)$. 
\end{lemma}
\begin{proof}
  To prove the lemma, we first make two sub-claims:
  \begin{claim}\label{clm:1}
    $\E[Gain(S_2 \cup S_3, \emptyset) - Gain(S_2 \cup S_3, A^{\G}(S_1))]=
    \sum_{j=1}^{n/2} a_j \frac{n/2}{n-j}$. 
  \end{claim}
  
  \begin{claim}\label{clm:2}
    $\E[Gain(S_1, \emptyset) - Gain(S_1, A^{\G}(S_1))] =
    \sum_{j=1}^{n/2} \left(a_j \frac{n/2-j}{n-j} + b_j \right)$
  \end{claim}

  Both these claims follow from the definitions of $a_j$ and $b_j$,
  along similar lines to the proof of Lemma~\ref{lem:wiai}.  
 We recall that $S_1$ is the first half of random permutation $\pi$ and $S_2 \cup S_3$ is its second half, therefore the two sets $S_1$ and $S_2 \cup S_3$ are identically distributed and symmetric.   
So it follows from
  Claim~\ref{clm:1} that there exists an allocation $A'$ of items in
  $S_2 \cup S_3$ that reduces the $Gain$ of items in $S_1$ in expectation by the same
  amount as Greedy's assignment of $S_1$ reduces the $Gain$ of $S_2
  \cup S_3$ in expectation. That is, $\E[Gain(S_1, \emptyset) - Gain(S_1, A')] =
  \sum_{j=1}^{n/2} a_j \frac{n/2}{n-j}$.

  \begin{eqnarray*}
    X & = & Gain\left(S_1, A^{\G}(S_1)\right) - Gain\left(S_1, A^{\G}(S_1) \cup \hat{A}\right) \\
    &\ge& Gain\left(S_1, A^{\G}(S_1)\right) - Gain\left(S_1, A^{\G}(S_1) \cup A'\right)\\
  &=& Gain(S_1, \emptyset) - Gain\left(S_1,  A^{\G}(S_1) \cup A'\right) - \left(Gain(S_1, \emptyset) - Gain\left(S_1, A^{\G}(S_1)\right)  \right)\\
  &\ge& Gain(S_1, \emptyset) - Gain\left(S_1, A'\right) - \left(Gain(S_1, \emptyset) - Gain\left(S_1,
      A^{\G}(S_1)\right)  \right)
  \end{eqnarray*}

where the first inequality follows from the definition of $\hat{A}$,
which is the allocation of $S_2 \cup S_3$ that maximally reduces the
$Gain$ of $S_1$ after Greedy's assignment of $S_1$. 
The last inequality follows from submodularity of valuation functions that implies $Gain(j,A)$ values to be monotone decreasing in set $A$ for any item $j$. 
The lemma now
follows from taking the expectation of above equation and applying the two Claims~\ref{clm:1} and \ref{clm:2}. 
\end{proof}

So far we showed how to set variables $g_i$ such that constraints $(3)$ is not violated. We now show the same setting of variables $g_i$ satisfies constraint $(4)$ as well. We note that by definition $\beta$ is equal to $\sum_{i=1}^n b_i$, therefore proving constraint $(4)$ is equivalent of proving Lemma~\ref{lem:b_and_g}. 


\begin{lemma}\label{lem:b_and_g}
  $\sum_{i=n/2+1}^n g_i \le  \sum_{i=n/2+1}^n b_i$.
\end{lemma}
\begin{proof}
  The proof of this lemma proceeds via two intermediate claims:
  \begin{claim}\label{clm:second-order}
    Let $f$ be a monotone submodular function defined on a universe
    $U$. For fixed sets $S, Z \subseteq U$, we define a set function
    $R(A)$ for $A \subseteq U \setminus (S \cup Z)$ as follows:
    \[ R(A) = \left(f(S \cup Z) - f(Z) \right) - \left(f(S \cup Z \cup
      A) - f(Z \cup A) \right). \]
    Intuitively, $R(A)$ captures the ability of $A$ to reduce the
    marginal gain from $S$. If $f$ is a second-order supermodular
    function, then $R$ is a submodular function.
  \end{claim}
  
  We defer the proof of this claim temporarily, as it follows from the
  definition of second-order supermodular functions in a straightforward
  way. 

  \begin{claim}\label{clm:yi+1}
    $\E[Y_{i+1}] \ge \frac{n-i}{n-i+1} \E[Y_i] - b_i$.
  \end{claim}
  \begin{proof}
    Claim~\ref{clm:second-order} implies that if each agent's
    valuation function is second-order supermodular, then the reduction
    in $Gain$ values of items in $S_1$ by allocating a set of items
    from $S_2 \cup S_3$ is a submodular function of this set of
    items, \emph{keeping previously assigned items fixed}. Now,
    consider $Y_i = Gain(S_1, A^{\G}_{i-1}) - Gain(S_1, A^{\G}_{i-1} \cup
    \hat{A}_{i+})$, which is the reduction in $Gain$ of items in $S_1$
    by allocating the last $n-i+1$ items according to $\hat{A}$,
    assuming Greedy allocated the first $i-1$ items. We
    compare this to $Gain(S_1, A^{\G}_{i-1}) - Gain(S_1, A^{\G}_{i-1} \cup
    \hat{A}_{(i+1)+})$, which is the reduction in $Gain$ of items in
    $S_1$ by allocating the last $n-i$ items according to $\hat{A}$,
    again assuming Greedy allocated the same first $i-1$
    items. Keeping the first $i-1$ items fixed, the last $n-i$ items
    are a random subset of the last $n-i+1$ items, and since the
    $Gain$ reduction function is submodular, we have (see
    e.g. \cite{FeigeMV11}): 

    \begin{eqnarray*}
      &&\E\left[Gain\left(S_1, A^{\G}_{i-1}\right) - Gain\left(S_1, A^{\G}_{i-1}
        \cup \hat{A}_{(i+1)+}\right) \right]\\
         &\ge& \frac{n-i}{n-i+1}
      \left( Gain\left(S_1, A^{\G}_{i-1}\right) - Gain\left(S_1, A^{\G}_{i-1}
        \cup \hat{A}_{i+} \right) \right) 
      = \frac{n-i}{n-i+1}Y_i
    \end{eqnarray*}
    
    Now, it is easy to complete the proof of the claim:
    \begin{eqnarray*}
      \E[Y_{i}] &\le& \frac{n-i+1}{n-i} \cdot \E\left[Gain(S_1, A^{\G}_{i-1}) -
        Gain(S_1, A^{\G}_{i-1} \cup \hat{A}_{(i+1)+}) \right] \\
      &\le& \frac{n-i+1}{n-i} \cdot \E\left[Gain(S_1, A^{\G}_{i-1}) -
        Gain(S_1, A^{\G}_i \cup \hat{A}_{(i+1)+}) \right] \\
      &\le& \frac{n-i+1}{n-i} \cdot \E\left[b_i + Gain(S_1, A^{\G}_{i}) -
        Gain(S_1, A^{\G}_i \cup \hat{A}_{(i+1)+}) \right] \\
      &=& \frac{n-i+1}{n-i} \left(b_i + \E\left[Y_{i+1}\right] \right)
    \end{eqnarray*}
  \end{proof}
  
  Given this claim, the lemma follows via straightforward
  computations:
  \begin{eqnarray*}
    \E[Y_{i+1}] \ge \frac{n-i}{n-i+1} \E[Y_i] - b_i 
    &\Rightarrow & \frac{\E[Y_{i+1}]}{n-i} \ge \frac{\E[Y_i]}{n-i+1} -
    \frac{b_i}{n-i}\\
    &\Rightarrow & \frac{\E[X]}{n/2} - g_{i+1} \ge \frac{\E[X]}{n/2} -
    g_i - \frac{b_i}{n-i}\\
    &\Rightarrow & g_{i+1} \le g_i + \frac{b_i}{n-i}\\
    &\Rightarrow & g_{i+1} \le \sum_{j=n/2+1}^i \frac{b_j}{n-j}
  \end{eqnarray*}

  where the third inequality follows from the definition of $g_i,
  g_{i+1}$, and the final inequality comes from repeatedly applying
  the preceding inequality together with $g_{n/2+1} = 0$. Therefore,
  \[ \sum_{i=n/2+1}^n g_i \le \sum_{i=n/2+1}^n \sum_{j=n/2+1}^{i-1}
  \frac{b_j}{n-j} = \sum_{j=n/2+1}^n b_j \sum_{i=j+1}^{n} \frac{1}{n-j} = \sum_{j=n/2+1}^n b_j.\]

  \noindent It now remains only to verify Claim~\ref{clm:second-order}:\\
  \begin{proofof}{Claim~\ref{clm:second-order}}
    To show that $R$ is a submodular function, we define $\Delta_R(e,
    A)$ to be $R(A \cup \{e\}) - R(A)$, and then argue that for sets
    $A, B$ such that $A \subseteq B$, we have $\Delta_R(e, A) \ge
    \Delta_R(e, B)$. By definition of $\Delta_R$, we have:

    \begin{eqnarray*}
    \Delta_R(e, A) &=& \left(f(S \cup Z) - f(S)\right) - \left(f(S
      \cup Z \cup A \cup \{e\}) - f(Z \cup A \cup
      \{e\}) \right) \\
    &\ & \qquad - \left[ 
      \left(f(S \cup Z) - f(S)\right) - \left(f(S
        \cup Z \cup A) - f(Z \cup A) \right) \right]\\
    &=& \left(f(Z \cup A \cup \{e\} - f(Z \cup A) \right) - 
    \left(f(S \cup Z \cup A \cup \{e\}) - f(S \cup Z \cup A) \right)\\
    &=& MG(Z \cup A, e) - MG(S \cup Z \cup A, e)
    \end{eqnarray*}
    where $MG(\mathcal{S}, e)$ denotes the marginal gain (under
    function $f$) of adding $e$
    to $\mathcal{S}$. If we use $A'$ to denote $Z \cup A$, we have:
    \[\Delta_R(e, A) = MG(A', e) - MG(S \cup A', e).\]
    Similarly, if $B'$ denotes $Z \cup B$, we have:
    \[\Delta_R(e, B) = MG(B', e) - MG(S \cup B', e).\]

    Now, obviously $A' \subseteq B'$, and from the definition of
    second-order supermodularity of $f$, we have $\Delta_R(e, A) \ge
    \Delta_R(e, B)$.
  \end{proofof}
  This completes the proof of Lemma~\ref{lem:b_and_g}.
\end{proof}

Recall that $\beta$ was defined as $\sum_{i=1}^n b_i$ which concludes the proof of constraint $(4)$:
\begin{corollary}\label{cor:beta_and_g}
  $\beta \ge \sum_{i=1}^{n/2} b_i + \sum_{n/2+1}^n g_i$.
\end{corollary}

\subsection{Analyzing the Factor-Revealing Linear Program}
\label{subsec:analyze_LP}

\input{analyzeMP.tex}

%% file: analyzeMP.tex
We want to lower bound the solution of $LP^{\beta}$ in this subsection.
To simplify the analysis, we will relax some constraints of
$\LP^{\beta}$ and analyze the relaxed linear program
$\LP^{\beta, \lambda}$ for some $\lambda \ge 1/2$ that we set later.
We note that $LP^{\beta, \lambda}$ is derived from $LP^{\beta}$ by removing two sets of constraints:
\begin{itemize}
\item For any $i > \lambda n$, the type $(2)$ constraint for $i$ is removed.
\item For any $n/2 < i \leq \lambda n$, the type $(3)$ constraint for $i$ is removed. 
\end{itemize}

\begin{center}
\begin{tabular}{|lll|} 
  \hline Linear program $\LP^{\beta, \lambda}$ && \\
  \hline Minimize $\sum_{i=1}^n w_i$ && \\ \hline 
  Subject to: && \\
  $w_i \geq b_i + a_i$ & $\forall 1 \leq i \leq n$ & $(1)$\\
  $w_i \geq 1/n  - \sum_{j=1}^{i-1} \frac{a_j}{n-j}$ & $\forall 1 \leq i
  \leq \lambda n$ & $(2)$\\
  $w_i \geq \frac{\sum_{j=1}^{n/2} \left(a_j \frac{j}{n-j} -
      b_j\right)}{n/2} - g_{i}$ & $\forall \lambda n < i \leq n$ & $(3)$\\ 
  $\beta \ge \sum_{i=1}^{n/2} b_i + \sum_{i=\frac{n}{2}+1}^{n}g_{i}$ && $(4)$ \vspace{0.5mm}
    \\ \hline
\end{tabular}
\end{center}

\begin{proposition}\label{prop:LPBetaLambdaLBLPBeta}
  For any $\lambda \ge 1/2$, the solution of $\LP^{\beta, \lambda}$ is
  at most the solution of $\LP^{\beta}$.
\end{proposition}

Thus, a lower bound on $\LP^{\beta, \lambda}$ will give us the desired
lower bound on $\LP^{\beta}$. In the following Theorem (proved at the end of this subsection), we lower bound the value of the optimum solution of $\LP^{\beta, \lambda}$.

\begin{theorem}\label{thm:MPBetaLambda}
  The solution of $\LP^{\beta, \lambda}$ is at least
  $\frac{1}{2}-\frac{(1-\lambda)^2}{2}+\frac{1-\lambda}{4} - \beta$
  for any $\beta \geq 0$, and $ \lambda > 1- (\sqrt{68}-8)$.
\end{theorem}

The following Corollary concludes the main goal of this subsection to lower bound $\LP^{\beta}$. 

\begin{corollary}\label{cor:LowerBoundMPBeta}
The solution of $\LP^{\beta}$ is at least $0.5312 - \beta$.
\end{corollary}
\begin{proof}
  Using Proposition~\ref{prop:LPBetaLambdaLBLPBeta}, Theorem~\ref{thm:MPBetaLambda}, and setting $\lambda = 0.754 >
  1-(\sqrt{68}-8)$, we achieve a lower bound of
  $\frac{1}{2}-\frac{(0.246)^2}{2}+\frac{0.246}{4} - \beta
  > 0.5312 - \beta$ on the solution of $LP^{\beta}$.
\end{proof}

Before elaborating on the formal proof Theoem~\ref{thm:MPBetaLambda}, we provide an overview of the analysis. 
To lower bound the optimum solution of $\LP^{\beta, \lambda}$, we take three main steps:
\begin{itemize}
\item We show there exists an optimum solution in which each $b_i$ is set to $0$, and therefore they all can be removed.
\item We then show one can assume that in the optimum solution, constraint $(2)$ is tight for any $i \leq \lambda n$, and also constraint $(3)$ is tight for any $i > \lambda n$. We also show that a choice of large enough $\lambda$ implies tightness of constraint $(1)$ for any $i$. 
\item After taking the first two steps, we have a simplified linear program with enough tight variables that the optimum solution can be explicitly computed. In particular, given the tightness of constraint $(2)$, one can inductively compute each $w_i$ as a function of $i$ and $n$ for any $i \leq \lambda n$. For higher values of $i$, tightness of constraint $3$ yields the exact value of $w_i + g_i$. Applying constraint $(4)$ bounds the sum of variables $g_i$ and consequently gives the value of optimum solution as a function of $\beta$.
\end{itemize}

We start by eliminating all variables
$\{b_i\}_{i=1}^n$.

\begin{lemma}\label{lem:b_zero}
  There exists an optimum solution for $\LP^{\beta, \lambda}$ in which
  $b_i = 0$ for all $1 \leq i \leq n$.
\end{lemma}
\begin{proof}
  Suppose $b_k > 0$ for some $k > n/2$: Set $b_k$ to 0, and note that
  we still have a feasible solution, since the first constraint still
  holds for $w_k$, and all other constraints are unaffected. Now,
  suppose $b_k > 0$ for some $k \le n/2$: Increase each $g_j$ by
  $\frac{b_k}{n/2}$ for each $j > n/2$, and then set $b_k$ to 0. 
Constraints $(1)$ and $(2)$ still hold since their right hand sides either remained unchanged or decreased. 
  The right hand side of  constraint $(3)$ for any $i >
  \lambda n$ is effected by the changes in variables $b_k$ and $g_i$. 
  Note that the coefficient of $b_k$ in this
  constraint is $-2/n$ and variable $g_i$
  was increased by exactly $\frac{b_k}{n/2}$. Therefore these two changes cancel each other out and  this lower bound  on $w_i$ is not raised.
  The constraint $(4)$ is also unchanged because in its right hand side, $n/2$ variables $g_i$ are each increased by $\frac{b_k}{n/2}$ and variable $b_k$ is decreased by this total amount.    
  Since all constraints are still satisfied,
  we can repeat this process for each $b_k > 0$ without raising the
  value of the objective function.
\end{proof}

We can now remove all variables $b_i$'s from $\LP^{\beta, \lambda}$. We
then prove that for all $i \le \lambda n$, constraint $(2)$ is a
tight lower bound for $w_i$, while for $i > \lambda n$, 
constraint $(3)$ is a tight lower bound. 
 
 \begin{lemma} \label{lem:tight_lower_bounds}
   There exists an optimum solution to $\LP^{\beta, \lambda}$ in which
constraints $(2)$ and $(3)$ are
   tight, and also $b_i = 0$ for all $1 \leq i \leq n$.
 \end{lemma}
 \begin{proof}
   Suppose this is not true; fix an optimal solution satisfying
   Lemma~\ref{lem:b_zero}. Let $k$ be the smallest index such that
   the inequality is not tight for $k$, and let $\delta$ be how much $w_k$ exceeds the lower
   bound from the relevant inequality. Decrease $w_k$ by
   $\delta$, and to ensure that we satisfy the constraint $w_k \ge b_k
   + a_k$, also decrease $a_k$ by $\delta$. To keep the rest of the
   constraints valid, for each $j > k$, we  increase $w_j$
   by $\delta/(n-k)$. 
   To see that these changes keep all the constraints satisfied, 
we note that variables $w_i$  do not appear on the right hand side of the constraints, so they do not 
change the lower bounds. The only change that effects the right hand sides is increment of $a_k$ which becomes relevant in constraint $(3)$ only for $k \leq n/2$. In this case, the increment does not exceed $\delta/(n-k)$, and the associated $w_i$ for this constraint is increased by this amount to keep it satisfied.   
So after these updates, we still have a feasible solution. Moreover, these updates do not increase the objective
   value. So far, we have argued that the inequality for $w_k$
   is now tight and we still have a feasible solution, but since some
   other constraint may no longer be tight, we have to show that this
   process terminates. To see this, note that we only affected
   constraints for $j > k$, so the smallest index for which the
   inequality is not tight has increased. Therefore, in at most $n$
   iterations, we find an optimum solution in which all inequalities
   are tight. We also note that in this process variables $b_i$ are intact which means we can assume they remain zero after all these updates given Lemma~\ref{lem:b_zero}.
 \end{proof}

 Up to this point, we have not used $\lambda$ at all, and all our
 claims hold for any $1/2 \leq \lambda \leq 1$. Our next lemma shows that 
 by choosing a large enough $\lambda > 1 - (\sqrt{68} - 8)$, one can assume that the constraints
 $w_i \geq a_i$ are tight for all $i$.

\begin{lemma}\label{lem:w_equal_a}
  If $\lambda > 1 - (\sqrt{68} - 8)$, any optimum
  point for $\LP^{\beta, \lambda}$ has $w_i=a_i$ for every $i$.
\end{lemma}
\begin{proof}
  Suppose $w_k > a_k$ for some $k > n/2$, simply set $a_k = w_k$. All
  constraints still hold, since $a_k$ appears only on the right hand
  side of constraint $(3)$ with a negative
  coefficient; therefore, increasing it keeps the solution feasible. Please also note that 
  for $k > n/2$, $a_k$ does not appear on the right hand side of constraint $(3)$. Therefore we can assume $a_k=w_k$ for any $k > n/2$.

Now we focus on the case $k \leq n/2$. 
  Suppose $w_k - a_k = \delta$ for some positive $\delta$ and $k \leq
  n/2$. Increase $a_k$ by $\delta$. To maintain feasibility, we need
  to adjust all $w_i$ and $a_i$ for $i >k$ as follows.  For each $k <
  i \leq \lambda n$, decrease both $a_i$ and $w_i$ by
  $\frac{\delta(n-i)}{(n-k)(n-k-1)}$, and for each $\lambda n < i \leq
  n$, increase $a_i$ and $w_i$ by $\frac{\delta k}{(n-k)n/2}$.

  \medskip 
  We first note that constraint $(1)$ remains satisfied since the $w_i$ and $a_i$ changes are synced.
  Constraint $(4)$ is also unaffected since we did not change its variables. 
  So to see the other constraints $(2)$ and $(3)$ are still feasible, we track the changes in
  their left and right hand sides. 
Clearly, these constraints are not affected for $i \leq k$.   
   For $k < i \leq \lambda n$ (associate with constraint $(2)$), $w_i$ is
  decreased by $\frac{\delta(n-i)}{(n-k)(n-k-1)}$, and we show that
  the right hand side of constraint $(2)$ is decreased
  accordingly.
For $i = k+1$, the decrease in both sides of the inequality is exactly
  $\frac{\delta}{n-k}$;  we prove it by
  induction for other values of $i$. The total change on the right
  hand side is:

  \[ -\frac{\delta}{n-k} + \sum_{j=k+1}^{i-1}
  \frac{\delta(n-j)}{(n-k)(n-k-1)} \times \frac{1}{n-j} =
  \frac{\delta}{(n-k)(n-k-1)} (-(n-k-1) + i-k-1) \]

  This sums up to a total decrease of
  $\frac{\delta(n-i)}{(n-k)(n-k-1)}$ which matches the decrement of the left hand side (variable $w_i$), and 
  shows constraint $(2)$ remains feasible, and 
  completes the claim for $i
  \leq \lambda n$.

  For $i > \lambda n$ (associated with constraint $(3)$), it is clear that the rightmost side of constraint $(3)$ is
  not increased by more than $\frac{\delta k}{(n-k)n/2}$ because the
  only term which increased in the right hand side
  of constraint $(3)$ corresponds to $a_k$, which increased by
  exactly this quantity. (Note that values of $a_i$ for $i > \lambda
  n$ were increased, but these either do not participate in the rightmost sides of constraints $(2)$ and $(3)$ or have negative coefficients there.) We also note that the leftmost side of constraint $(3)$ which is $w_i$ increased exactly this amount, so all the constraints remained feasible. 

  To conclude, we just need to show that all these changes did not
  increase the objective function $\sum_{i=1}^n w_{i}$. From values of
  $i \le \lambda n$, the total decrease in the objective is equal to:

  \[
  \sum_{i=k+1}^{\lambda n} \frac{\delta(n-i)}{(n-k)(n-k-1)} = 
  \frac{\delta[   (n-k-1)(n-k)/2 - (n-\lambda n -1)(n-\lambda n)/2 ]}{(n-k)(n-k-1)}
  \]

Since we are considering the case $k \leq n/2$, the above term is minimized at $k=n/2$, and is at least $\delta( 1/2 - (1-\lambda)^2/8)$ at that point. 

On the other hand, the total increase in the objective due to values
of $i > \lambda n$ is $(1-\lambda)n \times \frac{\delta
  k}{(n-k)n/2}$. This total is maximized at $k=n/2$, and is equal to
$2(1-\lambda)\delta$. For $\lambda > 1- (\sqrt{68}-8) \approx
0.7538$, the total increase in the objective is less than the total
decrease, which contradicts the optimality of the solution. So $w_i$
is indeed equal to $a_i$ for all $i$.
\end{proof}

Using the previous 3 lemmas, we can exactly compute the values of
each $w_i$, and therefore lower bound the solutions of $\LP^{\beta,
  \lambda}$ and $\LP^{\beta}$ as follows:

\begin{proofof}{Theorem~\ref{thm:MPBetaLambda}}
  Using Lemma~\ref{lem:b_zero} and Lemma~\ref{lem:tight_lower_bounds},
  we can eliminate all $b$ variables, and at the same time assume that
 constraints $(2)$ and $(3)$ of $\LP^{\beta,\lambda}$ 
  are tight.  To simplify further, we can use
  Lemma~\ref{lem:w_equal_a} to set each $a_i$ equal to $w_i$.

  To compute $w_i$ for each $i \leq \lambda n$, we use the tightness of
 constraint $(2)$ in $\LP^{\beta,\lambda}$. So we have:
  $$
  w_i = 1/n - \sum_{j=1}^{i-1} \frac{a_j}{n-j} = 1/n - \sum_{j=1}^{i-1} \frac{w_j}{n-j} 
  $$ 

  We can inductively prove that $w_i$ is $\frac{n-i}{(n-1)n}$ for any $i \leq \lambda n$. 

  Since we know $w_i$ for each $i \leq n/2$, we can compute the lower
  bound in constraint $(3)$ of $\LP^{\beta,\lambda}$. For each $i > \lambda
  n$, we know constraint $(3)$ is tight, giving:
  \begin{eqnarray*}
    w_i &=& \frac{\sum_{j=1}^{n/2} a_j \frac{j}{n-j}}{n/2} - g_{i} = \frac{\sum_{j=1}^{n/2} \frac{n-j}{(n-1)n} \times \frac{j}{n-j}}{n/2} - g_{i} 
    = \frac{\sum_{j=1}^{n/2} \frac{j}{(n-1)n}}{n/2} - g_{i} \\
    &=& \frac{ \frac{(n/2)(n/2+1)/2}{(n-1)n}}{n/2} - g_{i} 
    \geq \frac{1}{4n} - g_i
  \end{eqnarray*}

  Summing up these equations and applying constraint $(4)$, we
  have the following lower bound on the objective:
  \begin{eqnarray*}
    \sum_{i=1}^n w_i 
    &=& 
    \sum_{i=1}^{\lambda n} w_i + \sum_{i=\lambda n + 1}^n w_i 
    \geq 
    \frac{\sum_{i=1}^{\lambda n} n-i}{(n-1)n} + (1-\lambda) n (\frac{1}{4n}) - \sum_{i=\lambda n +1}^n g_i \\
    &\geq& 
    \frac{(n-1)n/2 - (n-\lambda n-1)(n-\lambda n)/2}{(n-1)n} + \frac{1-\lambda}{4} - \beta 
    \geq
    \frac{1}{2}-\frac{(1-\lambda)^2}{2}+\frac{1-\lambda}{4} - \beta
  \end{eqnarray*}

\end{proofof}

%% file: general-case.tex
In this section, we prove Theorem~\ref{thm:general-swm} that lower bounds the competitive 
ratio of Greedy for general submodular functions. 
For arbitrary submodular functions, we can no longer use the technique
of the previous section to argue that we obtain sufficient welfare
from each of the items in $S_2 \cup S_3$. Instead, here we use a
different constraint that provides a lower bound on $\sum_{i > 3n/4}
w_i$, which is Greedy's expected aggregate increase in welfare when assigning
$S_3$, the last quarter of items. We first need to set some notation for concatenation of multiple subsets (sequences) of items. 

\begin{definition}\label{def:sequence}
  We define $\la S_{1}, S_{2}, \cdots, S_{\ell} \rangle$ to be
  the concatenation of sequences of items $S_{1}, S_{2}, \cdots,
  S_{\ell}$. This definition can be similarly extended to any series of sequences of items.
   For any sequence $S_k$, we abuse notation and use
  $S_k$ to denote both a sequence of items and the set of these items,
  but the meaning will always be clear from context.
\end{definition}

We consider the simulated sequence of
items $S' = \la S_1, S_2, S_3, S_2, S_3, S_2 \ra$. In other words, we
will analyze how one could assign items if after the items in $S_1$
arrive, items in $S_2$ and $S_3$ arrive \emph{multiple} times. We
assign items of $S'$ using the following allocation $A' = A^{\G} (\la
S_1, S_2, S_3 \ra) \cup A^{\G}(\la S_2, S_3 \ra) \cup A^*(S_2)$. That
is, we first use Greedy to assign the items of $\la S_1, S_2, S_3
\ra$; then, we use the different allocation given by Greedy on $\la
S_2, S_3 \ra$ assuming nothing has been assigned so far (that is,
ignoring the previous allocation of Greedy on $\la S_1,S_2, S_3 \ra$).
Finally we use the optimum allocation for $S_2$. $A'$ is defined as
the union of these three allocations.  

To show that Greedy gets sufficient value from items in $S_3$, we will
lower bound the value of $A'$ on the last five segments $\la
S_2,S_3,S_2,S_3,S_2 \ra$.  This part of the value of $A'$ can be
formally written as $V(A') - V(A^{\G}(S_1))$\footnote{Recall from
  Definition~\ref{def:allocation} that $V(A)$ for some allocation $A$ denotes
  the total welfare or value of the allocation.}.
  The rest of this section consists of three major parts:
  
  \begin{itemize}[leftmargin=6mm]
  \item Lower bounding the marginal value of $A'$ on the last five segments by:
  \begin{eqnarray}\label{eq:marginalfivesegments}
\hspace{-20mm}
  \E[V(A')-V(A^{\G}(S_1))] \geq  \frac{1}{4} + \sum_{i=1}^{n/2} \left(\frac{i-n/4}{n-i}a_i - b_i \right) 
  + \sum_{i=n/2+1}^{3n/4} \frac{n/4}{n-i}a_i
  \end{eqnarray}
  This part of the proof is implied
  by noting that marginal value of allocation $A'$ on the last five segments is at least how much it reduces the $Gain$s of all items (Lemma~\ref{lem:GainReduction}), and then investigating the $Gain$ reductions of the sets $S_1$, $S_2$ and $S_3$ separately in Claims~\ref{claim:gainreductionSOne}, \ref{claim:gainreductionStwo} and \ref{claim:gainreductionSthree}.  

\item We then use Equation~\ref{eq:marginalfivesegments} and lower bound $\sum_{i > 3n/4} w_i$ which is the Greedy's expected welfare increase from items in $S_3$. This part is formalized in Lemma~\ref{lem:lower_bound_last_quarter_general_submod}. It is interesting to note that the lower bound terms of  both Equation~\ref{eq:marginalfivesegments} and consequently the one in Lemma~\ref{lem:lower_bound_last_quarter_general_submod} are in terms of $9n/4$ variables $\{a_i, b_i, w_i\}$ for $1 \leq i \leq 3n/4$. 

\item Finally, we can combine the lower bound of Equation~\ref{eq:marginalfivesegments}  with other Lemmas and conclude that competitive ratio of Greedy is at least: 

\begin{eqnarray}\label{eq:marginalall}
&\frac{1}{24}& +
  \sum_{i=1}^{n/2} \left((1+\frac{i-n/4}{6(n-i)})a_i +\frac{5}{6}
    b_i\right) + \sum_{i=n/2+1}^{3n/4} \left((\frac{5}{6} + \frac{n/4}{6(n-i)})a_i
    + \frac{5}{6}b_i \right) \nonumber \\
    &+& \sum_{i=1}^{3n/4} \frac{5}{6} (w_i - a_i - b_i)
\end{eqnarray}

which does not depend directly on the last quarter of  items because their marginal contributions in welfare has been already encoded in Equation~\ref{eq:marginalall}. The rest of the proof is analyzing Equation~\ref{eq:marginalall} and show that it is always at least $0.5052$ which completes the proof of Theorem~\ref{thm:general-swm}.
  \end{itemize}

  We now start by proving Equation~\ref{eq:marginalfivesegments} formally which is essentially lower bounding $\E[V(A')-V(A^{\G}(S_1))]$. 
  
  \begin{proofof}{Equation~\ref{eq:marginalfivesegments}}
  Using
Lemma~\ref{lem:GainReduction}, we know it suffices to measure how much this allocation reduces the $Gain$ values of all items. That is, $\E[V(A') -
V(A^{\G}(S_1))] \geq \E[Gain(N, A^{\G}(S_1)) - Gain(N, A')]$. Since $N
= S_1 \cup S_2 \cup S_3$, we can lower bound the desired quantity by
separately lower bounding $\E[Gain(S_k, A^{\G}(S_k)) - Gain(S_k, A')]$
for each of $k = 1, 2, 3$ which is done in Claims~\ref{claim:gainreductionSOne}, \ref{claim:gainreductionStwo} and \ref{claim:gainreductionSthree}.

\begin{claim}\label{claim:gainreductionSOne}
$\E[Gain(S_1, A^{\G}(S_1)) - Gain(S_1, A')] \ge \sum_{i=1}^{n/2} \frac{n/2}{n-i}a_i - \sum_{i=1}^{n/2} (\frac{n/2-i}{n-i}a_i + b_i)$. 
\end{claim}
\begin{proof}
  For any item $j$, $Gain(j,A^{\G}(\la S_2,S_3\ra))$ is at least
  $Gain(j,A')$ by submodularity and the fact that $A^{\G}(\la S_2,S_3
  \ra) \subseteq A'$.  So to lower bound
  $\E[Gain(S_1, A^{\G}(S_1)) - Gain(S_1, A')]$, it suffices to lower
  bound:\\
  $\E[Gain(S_1, A^{\G}(S_1)) - Gain(S_1, A^{\G}(\la S_2,S_3\ra ))] =$
  $$ \quad  \E[Gain(S_1, \emptyset) - Gain(S_1, A^{\G}(\la S_2,S_3
  \ra))] \ - \ \E[Gain(S_1, \emptyset) - Gain(S_1, A^{\G}(S_1)) ]$$
 
  It follows that $\E[Gain(S_1, \emptyset) - Gain(S_1, A^{\G}(\la
  S_2,S_3 \ra))]$ is equal to $\sum_{i=1}^{n/2} \frac{n/2}{n-i}a_i$
  from symmetry and the definition of $a_i$ variables. To see this,
  note that $\pi' = \la S_2, S_3, S_1 \ra$ is also a random
  permutation in which $\la S_2, S_3 \ra$ is the first half of items,
  and $S_1$ is the second half. Now we want to lower bound how much
  Greedy's allocation of first half items in the random permutation
  $\pi'$ reduces the $Gain$ variables of second half items of
  $\pi'$. We know that for any $i \leq n/2$, the $i$th item in $\pi'$
  reduces the $Gain$ variables of items that appear in positions
  greater than $i$ in $\pi'$ by $a_i$, and (as in
  Lemma~\ref{lem:wiai}), in expectation $\frac{n/2}{n-i}$ fraction of
  this reduction is associated with items in positions greater than
  $n/2$.  Similarly, we can consider the original random permutation
  $\pi = \la S_1, S_2, S_3 \ra$, and see that the allocation of items
  in $S_1$ reduces the $Gain$ variables of other items in $S_1$ by a
  total of $\sum_{i=1}^{n/2} \frac{n/2-i}{n-i}a_i + b_i$.  This proves
  the desired claim.
\end{proof}

\begin{claim}\label{claim:gainreductionStwo}
  $\E[Gain(S_2, A^{\G}(S_1)) - Gain(S_2, A')] = \frac{1}{4} - \sum_{i=1}^{n/2} \frac{n/4}{n-i}a_i$.
\end{claim}
\begin{proof}
  Since $S_2$ is a random quarter of all items, we have that
  $\E[Gain(S_2, \emptyset)] = \frac{1}{4}$, and since $A'$ concludes
  with the optimum allocation on items of $S_2$ ($A^*(S_2)$),
  $Gain(S_2, A')$ is zero.
  So $\E[Gain(S_2, A^{\G}(S_1)) - Gain(S_2, A')]$ is
  equal to $\frac{1}{4} - (\E[Gain(S_2, \emptyset) - Gain(S_2,
  A^{\G}(S_1))])$. The latter term can be exactly calculated from
  the fact that allocating the item in the $i$th position reduces the
  gains of subsequent items by $a_i$, and $n/4$ of the $n-i$ remaining
  items are in $S_2$; therefore, $\E[Gain(S_2,
  \emptyset) - Gain(S_2, A^{\G}(S_1))] = \sum_{i=1}^{n/2} \frac{n/4}{n-i}a_i$.
\end{proof}

\begin{claim}\label{claim:gainreductionSthree}
$\E[Gain(S_3,A^{\G}(S_1)) - Gain(S_3,A')] \ge \sum_{i=n/2+1}^{3n/4} \frac{n/4}{n-i}a_i$.
\end{claim}
\begin{proof}
  Since $A^{\G}(\la S_1,S_2 \ra)$ is a subset of $A'$, it suffices to
  lower bound 
  \[ \E[Gain(S_3,A^{\G}(S_1)) - Gain(S_3,A^{\G}(\la S_1,S_2
  \ra))]. \] 
  Similarly to the proof of the previous claim, this
  reduction in $Gain$ of items in $S_3$ from allocating the $i$th item for any $n/2 < i \leq 3n/4$
  is $\frac{n/4}{n-i}a_i$. This shows how much allocating items of $S_2$ reduce $Gain$ of items in $S_3$.
\end{proof}

From the three preceding claims,  we conclude Equation~\ref{eq:marginalfivesegments} which is:

  \begin{align*} 
  &\E[V(A')-V(A^{\G}(S_1))]  \geq \E[Gain(N, A^{\G}(S_1)) - Gain(N, A')] \\
  &\geq 
  \sum_{i=1}^{n/2} \frac{n/2}{n-i}a_i - \sum_{i=1}^{n/2} \left(
    \frac{n/2-i}{n-i}a_i + b_i \right) 
  + \frac{1}{4} - \sum_{i=1}^{n/2} \frac{n/4}{n-i}a_i
  + \sum_{i=n/2+1}^{3n/4} \frac{n/4}{n-i}a_i  \\
  &=  \frac{1}{4} + \sum_{i=1}^{n/2} \left(\frac{i-n/4}{n-i}a_i - b_i \right) 
  + \sum_{i=n/2+1}^{3n/4} \frac{n/4}{n-i}a_i \tag{\ref{eq:marginalfivesegments}}
  \end{align*}

where the first inequality holds because of Lemma~\ref{lem:GainReduction} and the second is derived by applying the three claims. 
This lower bound concludes the proof of Equation~\ref{eq:marginalfivesegments}.
\end{proofof}

Now that we have proved Equation~\ref{eq:marginalfivesegments}, the first main part of proof of Theorem~\ref{thm:general-swm} is complete. As the second main step, we apply Equation~\ref{eq:marginalfivesegments} to achieve a lower bound on the expected welfare gain of the last quarter of items.

\begin{lemma}\label{lem:lower_bound_last_quarter_general_submod}
  The expected increase in welfare that Greedy achieves for the last
  quarter of items $\sum_{i > 3n/4} w_i$ is at least $\frac{1}{24} +
  \sum_{i=1}^{n/2} (\frac{i-n/4}{6(n-i)}a_i - \frac{1}{6}b_i)$ $+
  \sum_{i=n/2+1}^{3n/4} \frac{n/4}{6(n-i)}a_i -
  \frac{1}{6}\sum_{i=n/2+1}^{3n/4} w_i$.
\end{lemma}
\begin{proof}
  We use $X$ to denote the lower bound obtained in Equation~\ref{eq:marginalfivesegments} on how much $A'$ increases the welfare from the
  rest of $S'$ after items in $S_1$ have already arrived. In other
  words, this is how much it increases its welfare by allocating $\la
  S_2, S_3, S_2, S_3, S_2 \ra$. We use $W_3 = \sum_{i \in S_3} w_i$ to
  denote how much $A'$ increases its welfare by allocating items in
  the first copy of $S_3$, and since items in the first copy are
  assigned greedily (in $A'$), their total increase in the welfare is
  at least as much as the increase in welfare $A'$ obtains by
  allocating the second copy $S_3$ (this is easy to verify for each
  item). So $A'$ in total increases its welfare from the two copies of
  $S_3$ by at most $2W_3$.

  Let $S'_2$ be the sequence $\la S_2, S_2, S_2 \ra$ derived from $S'$
  by removing the initial copy of $S_1$ and the two copies of
  $S_3$. Let $A'_2$ be the projection of allocation $A'$ on sequence
  $S'_2$; in other words, $A'_2$ is an allocation of sequence $S'_2$,
  and it is consistent with $A'$ on this sequence. It is clear that
  after the allocation $A^{\G}(S_1)$, allocating the three copies of
  $S_2$ using $A'_2$ increases the welfare by at least $X-2W_3$ since
  removing the two copies of $S_3$ will not reduce the welfare by more
  than $2W_3$.

  We now appeal to symmetry and switch the argument from $S_2$ to
  $S_3$. Given set $S_1$, $S_2$ is a random sequence of half of the
  items in $N \setminus S_1$, and so is $S_3$; that is, given sequence
  $S_1$, the two sequences $S_2$ and $S_3$ have the same distribution. Suppose $S_1$ is
  allocated by Greedy; as argued above, we know that after this
  allocation it is possible to allocate three consecutive copies of
  $S_2$ using $A'_2$ and increase the welfare by at least $X-2W_3$ in
  expectation. Therefore we can claim that after allocation
  $A^{\G}(S_1)$, it is possible to allocate three copies of $S_3$ and
  increase the welfare by the same amount of at least $X-2W_3$.  Formally, there
  exists an allocation $A'_3$ of sequence $S'_3=\la S_3, S_3, S_3 \ra$
  such that $\E[V(A^{\G}(S_1) \cup A'_3) - V(A^{\G}(S_1))]$ is at
  least $X-2W_3$.

  Now consider the sequence $\la S_1, S_2, S_3, S'_3 \ra$, which
  begins with the original sequence $S = \la S_1, S_2, S_3 \ra$ and
  then has three copies of $S_3$. For this sequence, consider the
  allocation $A^{Final} = A^{\G}(\la S_1,S_2,S_3 \ra) \cup
  A'_3$; that is, we first assign the original
  sequence $\la S_1, S_2, S_3 \ra $ according to Greedy, and then
  assign the three copies of $S_3$ using $A'_3$.  The increase in
  welfare by allocating the last three copies of $S_3$ in $A^{Final}$
  is at least $X-2W_3 - W_2-W_3$ where $W_2 = \sum_{i \in S_2}
  w_i$. Now there are four copies of $S_3$ in the sequence $\la S_1, S_2,
  S_3, S'_3 \ra$. Since the first copy of $S_3$ is allocated greedily in
  $A^{Final}$, its increase in welfare which is $W_3$ is at least as
  much as the increase in welfare of any the other copies of $S_3$ in
  $A^{Final}$, and hence also at least as much as the average of
  increase in welfare by these three copies. We conclude that:

  $$
  W_3 \geq \frac{X-2W_3 - W_2 - W_3}{3}
  $$  
  
  which implies $W_3 \geq \frac{X-W_2}{6}$. By definition of $X$ and $W_2$, we have: 
  $$
  W_3 \geq \frac{\frac{1}{4} + \sum_{i=1}^{n/2} (\frac{i-n/4}{n-i}a_i - b_i) 
    + \sum_{i=n/2+1}^{3n/4} \frac{n/4}{n-i}a_i - \sum_{i=n/2+1}^{3n/4} w_i}{6}
  $$  

  which concludes the proof. 
\end{proof}

 We now have all the building blocks we need to conclude the last step of proof of Theorem~\ref{thm:general-swm}. We provide a lower bound
on the total welfare obtained by Greedy in terms of $w_i, a_i,$ and
$b_i$ associated with the first $3n/4$ items. This lower bound is strong enough that we can consequently prove Greedy achieves a competitive ratio
better than $\frac{1}{2}$.

\begin{lemma}\label{lem:lower_bound_greedy_gain} 
  The expected welfare obtained by Greedy is at least $\frac{1}{24} +
  \sum_{i=1}^{n/2} \left((1+\frac{i-n/4}{6(n-i)})a_i +\frac{5}{6}
    b_i\right) + \sum_{i=n/2+1}^{3n/4} \left((\frac{5}{6} + \frac{n/4}{6(n-i)})a_i
    + \frac{5}{6}b_i \right) + \sum_{i=1}^{3n/4} \frac{5}{6} (w_i - a_i - b_i)$.
\end{lemma}

\begin{proof}
  We can write the expected increase in welfare by Greedy as
  $\sum_{i=1}^{n/2} w_i + \sum_{i=n/2+1}^{3n/4} w_i +
  \sum_{i=3n/4+1}^{n} w_i$. Using
  Lemma~\ref{lem:lower_bound_last_quarter_general_submod} to lower
  bound the last term $\sum_{i > 3n/4} w_i$, we conclude that the
  welfare is at least $\frac{1}{24} + \sum_{i=1}^{n/2}
  (\frac{i-n/4}{6(n-i)}a_i - \frac{1}{6}b_i)$ $+ \sum_{i=n/2+1}^{3n/4}
  \frac{n/4}{6(n-i)}a_i + \frac{5}{6}\sum_{i=n/2+1}^{3n/4} w_i +
  \sum_{i=1}^{n/2} w_i$. 

  Since the coefficient of each $w_i$ is positive (either
  $\frac{5}{6}$ or $1$) in this lower bound, we can apply
  Lemma~\ref{lem:GainReduction} to replace each $w_i$ with $a_i+b_i$,
  and add the sum $\sum_{i=1}^{3n/4} \frac{5}{6} (w_i - a_i - b_i)$
  while still having a valid lower bound. It suffices to merge the
  sums to conclude the claim of this lemma.
\end{proof} 

Using Lemma~\ref{lem:lower_bound_greedy_gain}, proving Theorem~\ref{thm:general-swm} is reduced to analyzing the minimum value expression \ref{eq:marginalall} can take based on the $9n/4$ variables $\{a_i, b_i, w_i\}_{i=1}^{3n/4}$. One can think about it as a linear program that we analyze in the rest:

\begin{proofof}{Theorem~\ref{thm:general-swm}}
  Using Lemma~\ref{lem:lower_bound_greedy_gain}, we know that for any
  instance the expected gain of greedy is at least:
  $$LB = \frac{1}{24} + \sum_{i=1}^{n/2}
  \left((1+\frac{i-n/4}{6(n-i)})a_i +\frac{5}{6} b_i\right)+
  \sum_{i=n/2+1}^{3n/4} \left((\frac{5}{6} + \frac{n/4}{6(n-i)})a_i +
    \frac{5}{6}b_i \right) + \sum_{i=1}^{3n/4} \frac{5}{6} (w_i - a_i - b_i)$$
 
  for some set of $\frac{9n}{4}$ numbers
  $\{w_i,a_i,b_i\}_{i=1}^{3n/4}$ all in $[0,1]$. From
  Lemmas~\ref{lem:GainReduction} and \ref{lem:wiai}, we know they
  satisfy the inequalities $w_i \geq a_i + b_i$, and $w_i \geq
  \frac{1}{n} - \sum_{j=1}^{i-1}\frac{a_j}{n-j}$ for any $1 \leq i
  \leq \frac{3n}{4}$. The minimum value of $LB$ among the feasible
  solutions of these linear constraints lower bounds the competitive
  ratio of Greedy algorithm. We prove that to minimize $LB$ given
  these linear constraints, one can assume w.l.o.g. that $w_i = a_i +
  b_i$ for any $i \leq \frac{3n}{4}$ as follows.  Suppose for some $i
  \leq \frac{3n}{4}$, the gap $\delta = w_i - (a_i + b_i)$ is
  positive. We increase $b_i$ by $\delta$. The value of $LB$ is intact
  because on one hand $\frac{5}{6}b_i$ appears in either the first or
  second summation in the lower bound, and $-\frac{5}{6}b_i$ appears
  in the third summation. The inequalities $w_i \ge a_i + b_i$ and
  $w_i \ge \frac{1}{n} - \sum_{j=1}^{i=1} \frac{a_j} {n-j}$ are also
  still satisfied, and $w_i$ is now equal to $a_i+b_i$. Performing
  this update for every constraint $w_i \geq a_i+b_i$ that is not
  tight, we can assume that $w_i = a_i + b_i$ for any $i \leq
  \frac{3n}{4}$. Therefore it suffices to lower bound the following
  simplified expression $LB'$:
  
  $$LB' = \frac{1}{24} + \sum_{i=1}^{n/2}
  \left((1+\frac{i-n/4}{6(n-i)})a_i +\frac{5}{6} b_i\right)+ \sum_{i=n/2+1}^{3n/4} \left((\frac{5}{6} + \frac{n/4}{6(n-i)})a_i +
    \frac{5}{6}b_i \right)$$
  for some set of $\frac{3n}{2}$ numbers $\{a_i,b_i\}_{i=1}^{3n/4}$
  all in $[0,1]$ with linear constraints  $a_i + b_i \geq
  \frac{1}{n} - \sum_{j=1}^{i-1}\frac{a_j}{n-j}$ for any $1 \leq i
  \leq \frac{3n}{4}$. 

  We find the minimum value of $LB'$ among all feasible solutions of
  these linear constraints. Our analysis is somewhat similar to that
  in Subsection~\ref{subsec:analyze_LP} for the parametrized linear
  program.  
  First, we prove that there exists an
  optimum solution (minimizing $LB'$, and satisfying the constraints)
  in which $a_i+b_i = \frac{1}{n} - \sum_{j=1}^{i-1}\frac{a_j}{n-j}$
  for any $1 \leq i \leq \frac{3n}{4}$. In other words, the linear
  constraints should be tight. Second, we show that there exists an
  optimal solution in which, furthermore, all $b_i$ values are
  $0$. This then allows us to explicitly find the values of $a_i$ in
  this optimal solution, and we can evaluate $LB'$ explicitly.

  We begin by showing that all the linear constraints are tight, and
  $b_i$ values are $0$; 
  we prove this by contradiction. Let
  $\{a^*_i,b^*_i\}_{i=1}^{3n/4}$ be an optimum solution of this linear
  program.
  Suppose $i$ is the smallest index for which either this constraint is not
  tight or $b^*_i$ is not zero. At first we consider the case that the constraint
  is not tight. Define $\delta > 0$ to be the gap $a^*_i+b^*_i
  -\left(\frac{1}{n} - \sum_{j=1}^{i-1}\frac{a^*_j}{n-j}\right)$. We
  change the solution as follows and claim that the new solution is
  still feasible and the objective function $LB'$ has not increased.

$$
a^{new}_j = 
\begin{cases} 
a^*_j &\mbox{if } j < i \\ 
a^*_j - \delta &\mbox{if } j = i \\ 
a^*_j + \frac{\delta(n-j)}{(n-i)(n-i-1)} &\mbox{if } j > i 
\end{cases}
$$ 

We keep all $b$ values intact, i.e. $b^{new}_j = b^*_j$ for each $1 \leq j \leq \frac{3n}{4}$. 
To prove the feasibility of solution is maintained, we should show that constraint $a_j + b_j \geq \left(\frac{1}{n} - \sum_{j'=1}^{j-1}\frac{a_{j'}}{n-j'}\right)$ still holds. 
For $j < i$, all variables are intact, and therefore this constraint still holds. For $j=i$, the constraint holds (and is tight now) by definition of  $\delta$. 
For $j > i$, the left hand side is increased by $\frac{\delta(n-j)}{(n-i)(n-i-1)}$, and the right hand side is changed by $\frac{\delta}{n-i} - \sum_{j'=i+1}^{j-1} \frac{\delta(n-j')}{(n-i)(n-i-1)}\frac{1}{n-j'} = \frac{\delta}{n-i} - \sum_{j'=i+1}^{j-1} \frac{\delta}{(n-i)(n-i-1)} = \frac{\delta}{(n-i)(n-i-1)} ((n-i-1) - (j-1-i)) = \frac{\delta (n-j)}{(n-i)(n-i-1)}$. Therefore both sides of the constraint change in the same way which means the constraint still holds, and the solution remains feasible. 
Now we prove that the objective function has not increased. There are two cases: a) $i \leq \frac{n}{2}$, and b) $i > \frac{n}{2}$. For $i \leq \frac{n}{2}$, the change in the objective function is equal to:

\begin{eqnarray*}
&-&\left(1+\frac{i-n/4}{6(n-i)}\right)\delta +
\sum_{j=i+1}^{n/2} \left(1+\frac{j-n/4}{6(n-j)}\right)\frac{\delta(n-j)}{(n-i)(n-i-1)} + 
\sum_{j=n/2+1}^{3n/4} \left(\frac{5}{6} + \frac{n/4}{6(n-j)}\right)\frac{\delta(n-j)}{(n-i)(n-i-1)} \\
&=& \frac{\delta}{6(n-i)(n-i-1)} 
\left(
-\left(\frac{23n}{4}-5i\right)(n-i-1) + \sum_{j=i+1}^{n/2} \left(\frac{23n}{4}-5j\right) + \sum_{j=n/2+1}^{3n/4} \left(\frac{21n}{4}-5j\right)
\right)\\
&=& \frac{\delta}{6(n-i)(n-i-1)} 
\left(
-\left(\frac{23n}{4}-5i\right)(n-i-1) + \frac{23n}{4}(n/2-i) + \frac{21n}{4}(n/4) -\sum_{j=i+1}^{3n/4} 5j 
\right)\\
&=& \frac{\delta}{6(n-i)(n-i-1)} 
\left(
 5i(n-i-1) - \frac{23n}{4}(n/2-1) + \frac{21n}{4}(n/4) - 5 \frac{(3n/4)(3n/4+1)}{2} +5 \frac{i(i+1)}{2}
\right)\\
&=& \frac{\delta}{6(n-i)(n-i-1)} 
\left(
 5i(n-(i+1)/2) - \frac{23n}{4}(n/2-1) + \frac{21n}{4}(n/4) - 5 \frac{(3n/4)(3n/4+1)}{2}
\right)
\end{eqnarray*}

If the above expression is positive for some $1 \leq i \leq \frac{n}{2}$, it should be positive for $i= n/2$ that maximizes the term $i(n-(i+1)/2)$. We note that the coefficient $\frac{\delta}{6(n-i)(n-i-1)} $ is positive for all values of $i$. For $i=\frac{n}{2}$, the above expression becomes: 

\begin{eqnarray*}
&& \frac{\delta}{6(n-i)(n-i-1)} 
\left(
 5\frac{n}{2}(n-(n/2+1)/2) - \frac{23n}{4}(n/2-1) + \frac{21n}{4}(n/4) - 5 \frac{(3n/4)(3n/4+1)}{2}
\right)\\
&=& \frac{\delta}{6(n-i)(n-i-1)} 
\left(
(\frac{15}{8} -\frac{23}{8} +\frac{21}{16} -\frac{45}{32})n^2
+ (-\frac{5}{4} + \frac{23}{4} - \frac{15}{8})n
\right)\\
&=& \frac{\delta}{6(n-i)(n-i-1)} 
\left(
(\frac{60-92+42-45}{32})n^2
+ (\frac{-10+46-15}{8})n
\right)\\
&=& \frac{\delta}{6(n-i)(n-i-1)} 
\left(
(\frac{-35}{32})n^2
+ (\frac{21}{8})n
\right)
\end{eqnarray*}

The above expression is negative for $n > 2$. Therefore these updates do not increase the objective function for any $i \leq \frac{n}{2}$. To prove the same claim for $i > \frac{n}{2}$, we write down the changes in the objective function:

\begin{eqnarray*}
&-&\left(\frac{5}{6}+\frac{n/4}{6(n-i)}\right)\delta +
\sum_{j=i+1}^{3n/4} \left(\frac{5}{6} + \frac{n/4}{6(n-j)}\right)\frac{\delta(n-j)}{(n-i)(n-i-1)} \\
&=& \frac{\delta}{6(n-i)(n-i-1)} 
\left(
-\left(\frac{21n}{4}-5i\right)(n-i-1) + \sum_{j=i+1}^{3n/4} \left(\frac{21n}{4}-5j\right)
\right)
\end{eqnarray*}

It is not hard to see that the above expression is always non-positive. We note that $\frac{21n}{4}-5i$ is greater than $\frac{21n}{4}-5j$ for each $j > i$, and the number of different $j$ indices in the sum is less than $n-i-1$. Therefore the expression inside the large parentheses is not positive, and its coefficient $\frac{\delta}{6(n-i)(n-i-1)} $ is positive. We conclude that the updates do not increase the objective function for any value of $i$. Therefore we can assume that the constraints are tight: $a^*_i + b^*_i = \frac{1}{n} - \sum_{j=1}^{i-1}\frac{a^*_j}{n-1}$. 

Now we show how to update the solution when $b^*_i$ is equal to some $\delta > 0$. 
We prove that without increasing the objective function we can set $b^*_i$ to zero.
We change the solution as follows: 

$$
a^{new}_j = 
\begin{cases} 
a^*_j &\mbox{if } j < i \\ 
a^*_j + \delta &\mbox{if } j = i \\ 
a^*_j - \frac{\delta(n-j)}{(n-i)(n-i-1)} &\mbox{if } j > i 
\end{cases}
$$ 

We set $b^{new}_i = 0$, and keep the rest of $b$ values intact. Similar to the previous updates, these changes maintain feasibility. We just need to prove that the objective function has not increased. Again we have two cases. If $i \leq \frac{n}{2}$, the change in the objective function is:

\begin{eqnarray*}
&&\left(\frac{1}{6}+\frac{i-n/4}{6(n-i)}\right)\delta 
-\sum_{j=i+1}^{n/2} \left(1+\frac{j-n/4}{6(n-j)}\right)\frac{\delta(n-j)}{(n-i)(n-i-1)} 
-\sum_{j=n/2+1}^{3n/4} \left(\frac{5}{6} + \frac{n/4}{6(n-j)}\right)\frac{\delta(n-j)}{(n-i)(n-i-1)} \\ 
&=& \frac{\delta}{6(n-i)(n-i-1)} 
\left(
\frac{3n}{4}(n-i-1) - \sum_{j=i+1}^{n/2} \left(\frac{23n}{4}-5j\right) - \sum_{j=n/2+1}^{3n/4} \left(\frac{21n}{4}-5j\right)
\right)\\
&=& \frac{\delta}{6(n-i)(n-i-1)} 
\left(
\frac{3n}{4}(n-i-1) - \frac{23n}{4}(n/2-i) - \frac{21n}{4}(n/4) +\sum_{j=i+1}^{3n/4} 5j 
\right)\\
&=& \frac{\delta}{6(n-i)(n-i-1)} 
\left(
 \frac{n}{4}(3n-5i-5 - 23n/2+23i-21n/4)+ 5 \frac{(3n/4)(3n/4+1)}{2} -5 \frac{i(i+1)}{2}
\right)\\
&=& \frac{\delta}{6(n-i)(n-i-1)} 
\left(
-\frac{65}{32}n^2
+ \frac{5}{8} n 
 +i\left(\frac{9n}{2} - \frac{5(i+1)}{2}\right)
\right)
\end{eqnarray*}

If the above expression is positive for some $1 \leq i \leq \frac{n}{2}$, it should be positive for $i= n/2$ that maximizes the term $i(\frac{9n}{2}-\frac{5(i+1)}{2})$. We note that the coefficient $\frac{\delta}{6(n-i)(n-i-1)} $ is positive for all values of $i$. For $i=\frac{n}{2}$, the above expression becomes: 

\begin{eqnarray*}
 \frac{\delta}{6(n-i)(n-i-1)} 
\left(
-\frac{65}{32}n^2
+ \frac{5}{8} n 
 +\frac{n}{2}\left(\frac{9n}{2} - \frac{5(n/2+1)}{2}\right)
\right)
= \frac{\delta}{6(n/2)(n/2-1)} (-\frac{13}{32}n^2
- \frac{5}{8} n)
\end{eqnarray*}

which is negative. So the updates do not increase the objective function for $i \leq \frac{n}{2}$. For $i > \frac{n}{2}$, the total change in the objective function is: 

\begin{eqnarray*}
&&\frac{n/4}{6(n-i)}\delta -
\sum_{j=i+1}^{3n/4} \left(\frac{5}{6} + \frac{n/4}{6(n-j)}\right)\frac{\delta(n-j)}{(n-i)(n-i-1)} \\ 
&=& \frac{\delta}{6(n-i)(n-i-1)} 
\left(
\frac{n}{4}(n-i-1) - \sum_{j=i+1}^{3n/4} \left(\frac{21n}{4}-5j\right)
\right)\\
&=& \frac{\delta}{6(n-i)(n-i-1)} 
\left(
\frac{n}{4}(n-i-1) - \frac{n}{4} \times \frac{21n}{4} + \frac{5}{2} \times \frac{3n}{4}\left(\frac{3n}{4}+1\right) - \frac{5}{2}i(i+1)
\right) \\
&=& \frac{\delta}{6(n-i)(n-i-1)} 
\left(
(\frac{1}{4} - \frac{21}{16} + \frac{45}{32})n^2 +
(-\frac{1}{4} + \frac{15}{8})n +
-(\frac{n}{4} + \frac{5}{2}(i + 1))i
\right)
\end{eqnarray*}

If the above expression is positive for some $ i > \frac{n}{2}$, it should be positive for $i= \frac{n}{2}+1$ that minimizes the term $i(\frac{n}{4}-\frac{5(i+1)}{2})$. For $i=\frac{n}{2}+1$, the above expression becomes: 

\begin{eqnarray*}
&&\frac{\delta}{6(n-n/2-1)(n-n/2-2)} 
\left(
\frac{11}{32}n^2 +
\frac{13}{8}n
-\left(\frac{n}{4} + \frac{5}{2}\left(\frac{n}{2} + 2\right)\right)\left(\frac{n}{2}+1\right)
\right) \\
&=& \frac{\delta}{6(n/2-1)(n/2-2)} 
\left(
-\frac{13}{32}n^2
-\frac{19}{8}n
-5
\right)
\end{eqnarray*}

which is negative for any value of $n$.

We note that these updates only change the variables with higher indices keeping $a^*_j$, and $b^*_j$ intact for $j < i$.   
Therefore after performing these update operations iteratively for at most $n$ times ($i=1, 2, \cdots, n$), we make sure that all constraints are tight, and all $b^*$ values are zero. 
Now we can assume that there exists an optimum solution $\{a^*_i, b^*_i\}_{i=1}^{3n/4}$ of the linear program that lower bounds the competitive ratio of Greedy in which $b^*_i$ is zero for any $i$, and $a^*_i = \frac{1}{n} - \sum_{j=1}^{i-1} \frac{a^*_j}{n-j}$ (the constraints are tight). We can recursively find $a$ values, and imply that $a^*_i$ is equal to $\frac{n-i}{(n-1)n}$. The objective function at this optimum point is equal to:

\begin{eqnarray*}
LB' &=& \frac{1}{24} + \sum_{i=1}^{n/2} \left(\left(1+\frac{i-n/4}{6(n-i)}\right)a^*_i +\frac{5}{6} b^*_i\right)+ \sum_{i=n/2+1}^{3n/4} \left(\left(\frac{5}{6} + \frac{n/4}{6(n-i)}\right)a^*_i + \frac{5}{6}b^*_i \right) \\
&=& \frac{1}{24} + \sum_{i=1}^{n/2} \left(\left(1+\frac{i-n/4}{6(n-i)}\right)a^*_i\right)+ \sum_{i=n/2+1}^{3n/4} \left(\left(\frac{5}{6} + \frac{n/4}{6(n-i)}\right)a^*_i \right)\\
&=& \frac{1}{24} + \sum_{i=1}^{n/2} \left(\frac{\frac{23n}{4}-5i}{6(n-i)} \times\frac{n-i}{(n-1)n}\right)
+ \sum_{i=n/2+1}^{3n/4} \left(\frac{\frac{21n}{4}-5i}{6(n-i)} \times \frac{n-i}{(n-1)n}\right)\\
&=& \frac{1}{24} + \frac{1}{6(n-1)n} 
\left(
\frac{23n^2}{8} + \frac{21n^2}{16} - 5 \sum_{i=1}^{3n/4} i
\right) 
= \frac{1}{24} + \frac{\frac{89n^2}{32} - \frac{15n}{8}}{6(n-1)n} >
 \frac{1}{24} + \frac{\frac{89n}{32}}{6n} \\
 &=& \frac{1}{24} + \frac{89}{192} \approx 0.5052
\end{eqnarray*}

Therefore Greedy increases the welfare by at least $0.5052$ fraction
of the optimum offline solution in expectation.
\end{proofof}

%% file: swm-greedy-Dec-2017.bbl
\begin{thebibliography}{10}

\bibitem{AGKM11}
{\sc Gagan Agarwal, Gagan Goel, Chinmay Karande, and Aranyak Mehta}, {\em
  Online vertex-weighted bipartite matching and single-bid budgeted
  allocation}, in SODA, SIAM, 2011.

\bibitem{AWY09}
{\sc Shipra Agrawal, Zizhuo Wang, and Yinyu Ye}, {\em A dynamic near-optimal
  algorithm for online linear programming}, Computing Research Repository,
  (2009).

\bibitem{AHL-Prophet}
{\sc Saeed Alaei, MohammadTaghi Hajiaghayi, and Vahid Liaghat}, {\em Online
  prophet-inequality matching with applications to ad allocation}, in
  Proceedings of the 13th ACM Conference on Electronic Commerce, ACM, 2012,
  pp.~18--35.

\bibitem{buchbinder-jain-naor}
{\sc Niv Buchbinder, Kamal Jain, and Joseph~Seffi Naor}, {\em Online
  primal-dual algorithms for maximizing ad-auctions revenue}, in ESA, Springer,
  2007, pp.~253--264.

\bibitem{devanur-hayes}
{\sc Nikhil Devanur and Thomas Hayes}, {\em The adwords problem: Online keyword
  matching with budgeted bidders under random permutations}, in EC, 2009,
  pp.~71--78.

\bibitem{DHKMY13}
{\sc Nikhil~R. Devanur, Zhiyi Huang, Nitish Korula, Vahab~S. Mirrokni, and Qiqi
  Yan}, {\em Whole-page optimization and submodular welfare maximization with
  online bidders}, in ACM Conference on Electronic Commerce, 2013,
  pp.~305--322.

\bibitem{DJSW11}
{\sc Nikhil~R. Devanur, Kamal Jain, Balasubramanian Sivan, and Christopher~A.
  Wilkens}, {\em Near optimal online algorithms and fast approximation
  algorithms for resource allocation problems}, in EC, ACM, 2011, pp.~29--38.

\bibitem{DevanurSA12}
{\sc Nikhil~R Devanur, Balasubramanian Sivan, and Yossi Azar}, {\em
  Asymptotically optimal algorithm for stochastic adwords}, in Proceedings of
  the 13th ACM Conference on Electronic Commerce, ACM, 2012, pp.~388--404.

\bibitem{DNS05}
{\sc Shahar Dobzinski, Noam Nisan, and Michael Schapira}, {\em Approximation
  algorithms for combinatorial auctions with complement-free bidders}, in STOC,
  2005, pp.~610--618.

\bibitem{DS06}
{\sc Shahar Dobzinski and Michael Schapira}, {\em An improved approximation
  algorithm for combinatorial auctions with submodular bidders}, in SODA, 2006,
  pp.~1064--1073.

\bibitem{FeigeMV11}
{\sc Uriel Feige, Vahab~S Mirrokni, and Jan Vondrak}, {\em Maximizing
  non-monotone submodular functions}, SIAM Journal on Computing, 40 (2011),
  pp.~1133--1153.

\bibitem{FHKMS10}
{\sc Jon Feldman, Monika Henzinger, Nitish Korula, Vahab~S. Mirrokni, and Cliff
  Stein}, {\em Online stochastic packing applied to display ad allocation}, in
  ESA, Springer, 2010, pp.~182--194.

\bibitem{FKMMP09}
{\sc J.~Feldman, N.~Korula, V.~Mirrokni, S.~Muthukrishnan, and M.~Pal}, {\em
  Online ad assignment with free disposal}, in WINE, 2009.

\bibitem{FMMM09}
{\sc Jon Feldman, Aranyak Mehta, Vahab Mirrokni, and S.~Muthukrishnan}, {\em
  Online stochastic matching: Beating 1 - 1/e}, in FOCS, 2009, pp.~117--126.

\bibitem{NWF78b}
{\sc M.~L. Fisher, G.~L. Nemhauser, and L.~A. Wolsey}, {\em An analysis of the
  approximations for maximizing submodular set functions {II}}, Mathematical
  Programming Study, 8 (1978), pp.~73--87.

\bibitem{GM08}
{\sc Gagan Goel and Aranyak Mehta}, {\em Online budgeted matching in random
  input models with applications to adwords}, in SODA, 2008, pp.~982--991.

\bibitem{HMZ11}
{\sc B.~Haeupler, V.~Mirrokni, and M.~ZadiMoghaddam}, {\em Online stochastic
  weighted matching: Improved approximation algorithms}, in WINE, 2011,
  pp.~170--181.

\bibitem{kp-balance}
{\sc Bala Kalyanasundaram and Kirk~R. Pruhs}, {\em An optimal deterministic
  algorithm for online b -matching}, Theoretical Computer Science, 233 (2000),
  pp.~319--325.

\bibitem{KPV13}
{\sc Michael Kapralov, Ian Post, and Jan Vondr{\'a}k}, {\em Online submodular
  welfare maximization: Greedy is optimal}, in SODA, 2013, pp.~1216--1225.

\bibitem{KMT11}
{\sc Chinmay Karande, Aranyak Mehta, and Pushkar Tripathi}, {\em Online
  bipartite matching with unknown distributions}, in STOC, 2011, pp.~587--596.

\bibitem{KVV}
{\sc R.M. Karp, U.V. Vazirani, and V.V. Vazirani}, {\em An optimal algorithm
  for online bipartite matching}, in STOC, 1990, pp.~352--358.

\bibitem{Kesselheim14}
{\sc Thomas Kesselheim, Andreas T{\"o}nnis, Klaus Radke, and Berthold
  V{\"o}cking}, {\em Primal beats dual on online packing lps in the
  random-order model}, in Proceedings of the 46th Annual ACM Symposium on
  Theory of Computing, ACM, 2014, pp.~303--312.

\bibitem{LLN}
{\sc Lehman, Lehman, and N.~Nisan}, {\em Combinatorial auctions with decreasing
  marginal utilities}, Games and Economic Behaviour,  (2006), pp.~270--296.

\bibitem{PLST12}
{\sc Renato~Paes Leme, Vasilis Syrgkanis, and {\'E}va Tardos}, {\em Sequential
  auctions and externalities}, in Proceedings of the Twenty-Third Annual
  ACM-SIAM Symposium on Discrete Algorithms, SIAM, 2012, pp.~869--886.

\bibitem{MY11}
{\sc Mohammad Mahdian and Qiqi Yan}, {\em Online bipartite matching with random
  arrivals: A strongly factor revealing lp approach}, in STOC, 2011,
  pp.~597--606.

\bibitem{MSVV}
{\sc Aranyak Mehta, Amin Saberi, Umesh Vazirani, and Vijay Vazirani}, {\em
  Adwords and generalized online matching}, J.~ACM, 54 (2007), p.~22.

\bibitem{MOS11}
{\sc Vahideh~H. Menshadi, Shayan Oveis~Gharan, and A.~Saberi}, {\em Online
  stochastic matching: Online actions based on offline statistics}, in SODA,
  2011, pp.~1285--1294.

\bibitem{MOZ12}
{\sc Vahab~S. Mirrokni, Shayan Oveis~Gharan, and Morteza ZadiMoghaddam}, {\em
  Simultaneous approximations of stochastic and adversarial budgeted allocation
  problems}, in SODA, 2011, pp.~1690--1701.

\bibitem{MSV08}
{\sc Vahab~S. Mirrokni, Michael Schapira, and Jan Vondr{\'a}k}, {\em Tight
  information-theoretic lower bounds for welfare maximization in combinatorial
  auctions}, in ACM Conference on Electronic Commerce, 2008, pp.~70--77.

\bibitem{MirrokniZ15}
{\sc Vahab~S. Mirrokni and Morteza Zadimoghaddam}, {\em Randomized composable
  core-sets for distributed submodular maximization}, in Proceedings of the
  Forty-Seventh Annual {ACM} on Symposium on Theory of Computing, {STOC} 2015,
  Portland, OR, USA, June 14-17, 2015, 2015, pp.~153--162.

\bibitem{MolinaroRavi}
{\sc Marco Molinaro and R~Ravi}, {\em The geometry of online packing linear
  programs}, Mathematics of Operations Research,  (2013).

\bibitem{NWF78}
{\sc G.~L. Nemhauser, L.~A. Wolsey, and M.~L. Fisher}, {\em An analysis of
  approximations for maximizing submodular set functions. {I}}, Math.
  Programming, 14 (1978), pp.~265--294.

\bibitem{SyrgkanisT12}
{\sc Vasilis Syrgkanis and Eva Tardos}, {\em Bayesian sequential auctions}, in
  Proceedings of the 13th ACM Conference on Electronic Commerce, ACM, 2012,
  pp.~929--944.

\bibitem{SyrgkanisT13}
\leavevmode\vrule height 2pt depth -1.6pt width 23pt, {\em Composable and
  efficient mechanisms}, in Proceedings of the 45th annual ACM symposium on
  Symposium on theory of computing, ACM, 2013, pp.~211--220.

\bibitem{VeeVS}
{\sc Erik Vee, Sergei Vassilvitskii, and Jayavel Shanmugasundaram}, {\em
  Optimal online assignment with forecasts}, in EC, 2010, pp.~109--118.

\bibitem{V08}
{\sc J.~Vondrak}, {\em Optimal approximation for the submodular welfare problem
  in the value oracle model}, in STOC, 2008, pp.~67--74.

\end{thebibliography}
